\documentclass[12pt]{article}

\usepackage{amsmath,
			amsthm,
			epsfig,
			amsfonts,
			color,
			booktabs,
			dsfont,
			graphicx,
			subfigure,
			float,
			mathrsfs,
			lscape,
			tikz}

\usepackage[round]{natbib}
\usepackage[nohead]{geometry}
\usepackage[singlespacing]{setspace}
\usepackage[hang,flushmargin]{footmisc}
\usepackage[colorlinks,
            linkcolor=blue,
            anchorcolor=blue,
            citecolor=blue]{hyperref}

\theoremstyle{remark}
\newtheorem{theorem}{{\sc Theorem}}
\newtheorem{prop}{Proposition}

\newtheorem{lemma}{{\sc Lemma}}
\newtheorem{remark}{{\sc Remark}}

\def\marginnote#1{\setbox0=\vtop{\hsize4pc \small\raggedright				\noindent\baselineskip9pt \rightskip=0.5pc plus 1.5pc #1}				\leavevmode\vadjust{\dimen0=\dp0 \kern-\ht0\hbox{\kern-4.00pc			\box0}\kern-\dimen0}}
\def\boxit#1{\vbox{\hrule\hbox{\vrule\kern6pt
          \vbox{\kern6pt#1\kern6pt}\kern6pt\vrule}\hrule}}

\newcommand{\bSigma}{\boldsymbol{\Sigma}}

\newcommand{\suit}[1]{\left(#1\right)}
\newcommand{\abs}[1]{\left\vert#1\right\vert}
\newcommand{\bra}[1]{\left[#1\right]}
\newcommand{\set}[1]{\left\{#1\right\}}
\newcommand{\df}{asymptotically unbiased distributed estimator }
\newcommand{\ldf}{Asymptotically unbiased distributed estimator }
\newcommand{\dfno}{asymptotically unbiased distributed estimator}

\begin{document}

\centerline{\large\bf Adapting the Hill estimator to distributed inference:  }
\vspace{2pt} 
\centerline{\large\bf dealing
with the bias}
\vspace{.4cm}

\centerline{Liujun Chen$^1$, Deyuan Li$^1$ and Chen Zhou$^{2,3}$} 

\begin{abstract}
The distributed Hill estimator is a divide-and-conquer algorithm for estimating the extreme value index when data are stored in multiple machines. In applications, estimates based on the distributed Hill estimator can be sensitive to the choice of the number of the exceedance ratios used in each machine. Even when choosing the number at a low level, a high asymptotic bias may arise. We overcome this potential drawback by designing a bias correction procedure for the distributed Hill estimator, which adheres to the setup of distributed inference. The asymptotically unbiased distributed estimator we obtained, on the one hand, is applicable to  distributed  stored data, on the other hand, inherits all known advantages of bias correction methods in extreme value statistics.

\end{abstract}

\noindent%

\medskip

\let\thefootnote\relax\footnotetext{Liujun Chen}
\footnotetext{ljchen19@fudan.edu.cn}
\footnote{}
\footnotetext{Deyuan Li}
\footnotetext{deyuanli@fudan.edu.cn}
\footnote{}
\footnotetext{Chen Zhou}
\footnotetext{zhou@ese.eur.nl}

\footnotetext{\vspace{3ex}}
\footnote{\noindent $^1$\quad School of Management, Fudan University, 220 Handan Road, Shanghai 200433, P.R. China.}
\footnotetext{}
\footnote{$^2$\quad Erasmus School of Economics, Erasmus University Rotterdam, P.O. Box 1738, 3000DR Rotterdam, The
Netherlands.}
\footnotetext{}
\footnote{$^3$\quad Economic Policy and Research Division, Bank of The Netherlands, P.O. Box 
98, 1000AB Amsterdam, The Netherlands.}

 \section{Introduction}

Consider a distribution function $F$ which belongs to the maximum domain of attraction of an extreme value distribution with a positive {\it extreme value index} $\gamma>0$,  that is,
$$
	\lim_{t\to \infty}\frac{U(tx)}{U(t)} = x^{\gamma}, \ x>0,
$$
where $U(t):=F^{\leftarrow}\suit{1-1/t}$ with $t>1$,
 and $ ^{\leftarrow}$ denotes the left-continuous inverse function. Such a distribution is also called a heavy-tailed distribution, where the  extreme value index  governs the tail of the distribution.  Estimating the extreme value index is a key step for making statistical inference on the  tail behaviour of $F$. Various methods have been proposed to estimate the extreme value index, such as the Hill estimator \citep{hill1975simple} , the maximum likelihood estimator \citep{smith1987estimating,drees2004maximum,zhou2009existence}  and the moment estimator \citep{dekkers1989moment}. 

Conducting extreme value analysis often requires large datasets in order to select extreme observations in the tail. 
Such datasets may be stored in multiple machines and cannot be combined into one dataset due to data privacy issue. For example, datasets collected in industries such as banking and  healthcare require high level consumer privacy and cannot be shared across different organizations. Another potential situation is that some massive datasets cannot be processed by a single computer due to internet traffic or memory constraints.  {\it Distributed inference} refers to the statistical problem of analyzing data stored in multiple machines. It often  requires a  divide-and-conquer (DC) algorithm. In a DC algorithm, one calculates 
statistical estimators on each machine in parallel and then communicates  them to a central machine. The final estimator is obtained on the central machine,
often  by a simple average; see, for example,  \citet{li2013statistical} for kernel density estimation, \citet{fan2019distributed} for principal component analysis, \citet{volgushev2019distributed} for quantile regression.

In this paper, we aim at estimating the extreme value index in the distributed inference context.
Assume that independent and identically distributed (i.i.d.) observations $X_1,\dots,X_N$ drawn from $F$  are stored in  $m$  machines  with $n$ observations on each machine, i.e. $N=mn$. In the context of distributed inference, we assume that only limited (finite) number of results can be transmitted from each machine to the central machine.
As a result, we cannot apply statistical procedures to the oracle sample, i.e., the hypothetically combined dataset $\set{X_1,\dots,X_N}$.

 \citet*{chen2021distributed}  proposes the  distributed Hill estimator to estimate the extreme value index $\gamma$. On each machine, the  Hill estimator is applied and then transmitted to the central machine.  On the central machine,  the average of the Hill estimates collected   from the $m$ machines are calculated. 
Let $M_j^{(1)}\ge \cdots \ge M_j^{(n)}$ denote the order statistics of the observations on machine $j$ for $j=1,\dots,m$. Then the Hill estimator on machine $j$ can be constructed by using the top $k$ exceedance ratios $M_j^{(i)}/M_j^{(k+1)}, i=1,\dots,k$, as 
$$
\hat{\gamma}_{j,k}=\frac{1}{k}\sum_{i=1}^k \suit{\log M_j^{(i)}-\log M_j^{(k+1)}}, \quad j=1,\dots,m.
$$
The distributed Hill estimator  is  defined as
$$
	\hat{\gamma}_{DH,k}:=\frac{1}{m}\sum_{j=1}^m \hat{\gamma}_{j,k}=\frac{1}{m}\sum_{j=1}^m  \frac{1}{k}\sum_{i=1}^k \suit{\log M_j^{(i)}-\log M_j^{(k+1)}}.
$$

\citet{chen2021distributed} studies the asymptotic behaviour of the distributed Hill estimator  and shows  sufficient conditions under which the distributed Hill estimator possesses the oracle property: its speed of convergence and asymptotic distribution coincides with the oracle Hill estimator. Here, the oracle Hill estimator is the Hill estimator using  the top  $km$ exceedance ratios of the oracle sample $\set{X_1,\dots,X_N}$, i.e. $\hat{\gamma}=l^{-1}\sum_{i=1}^{l}\suit{\log M^{(i)}-\log M^{(l+1)}}$, where $l=km$ and  $M^{(1)} \ge \cdots \ge M^{(N)}$ are the order statistics of the oracle sample $\set{X_1,\dots,X_N}$.  The choice of $l=km$ is in line with the standard distributed inference literature. Note that the oracle property compares the distributed estimator to the oracle estimator when the two estimators are constructed based on the same sample size. Different from standard statistics, extreme value statistics use observations in the tail only, for example, the Hill estimator is based on the exceedance ratios. Therefore, the oracle property for the Hill estimator is meaningful only if the distributed estimator and the oracle estimator are constructed based on the same number of exceedance ratios.

In applications with  finite sample size, one important tuning parameter in the Hill estimator is  the number of exceedance ratios $l$ used in the estimation.  Recall that the maximum domain of attraction condition is a limiting relation instead of an exact model, it provides only an approximation to the tail of a distribution. Consequently, the number of exceedance ratios used in the estimation, $l$, is related to the asymptotic bias in the limit distribution of the Hill estimator. This differs from classical statistics where bias often vanishes sufficiently fast as sample size tending to infinity. More specifically, the choice leads to a bias-variance tradeoff: with a low level of $l$, the estimation variance is at a high level; by increasing the level of $l$, the estimation variance is reduced but the estimation bias may arise. For the distributed Hill estimator $\hat{\gamma}_{DH,k}$, this issue is regarding the choice of $k$ on each machine. One needs to balance the number of exceedance ratios ($k$) with the number of machines ($m$), in order to control the total bias in the distributed estimator.   In addition, recall that the effective number of exceedance ratios involved in $\hat{\gamma}_{DH,k}$ is $km$. As $k$ increases by $1$,  the effective number of exceedance ratios will increase by $m$.  Thus, the performance of $\hat{\gamma}_{DH,k}$ is very sensitive to the choice of $k$.
If $m$ is large, with even a low level of $k$, the asymptotic bias may be  at a high level which may not be  acceptable  in  applications.

In existing extreme value statistics literature, there are two types of solutions for selecting the number of exceedance ratios in the estimation. 
The first stream of literature aims at finding the optimal level   that balances the asymptotic bias and variance, see e.g. \citet{danielsson2001using} and \citet{guillou2001diagnostic}.
The second stream of literature corrects the bias  and eventually  allows for choosing a high level of the number of exceedance ratios, see e.g. \citet{ivette2008tail} and \citet{de2016adapting}. In applications, if the sample size is large, the bias correction methods are preferred since they possess at least two advantages. First, bias correction methods allow for choosing a higher level of  the number of exceedance ratios than that used for the original  estimator, which results in also a lower level of variance. Second, bias correction methods lead to estimates that are less sensitive to the choice of the number of exceedance ratios.

In this paper, we shall adapt the distributed Hill estimator such that it is suitable for finite sample applications. More specifically, we introduce  a bias correction procedure for estimating the extreme value index,  without compromising the distributed inference setup.  Notice that  existing bias correction methods often rely on estimating a second order parameter and a second order scale function as given in \eqref{SOC} below. Such an estimation again requires the oracle sample which is infeasible in the context of distributed inference.  Therefore, we resort to a different approach, sticking to the requirement that only limited (fixed) number of results can be transmitted from each machine to the central machine. In such a way, the resulting estimator is not only asymptotically unbiased, but also in the same spirit of a DC algorithm. We name it as ``\dfno'' for the extreme value index.
The \dfno, on the one hand, is applicable to  distributed stored data, on the other hand, inherits the advantages of bias correction methods in extreme value statistics.  

 We remark that the requirement of transmitting limited (fixed) number of results from each machine to the central machine is in line with the privacy concern in practice. Consider a practical example where various insurance companies would not like to share their client level claim data, but would nevertheless be willing to collaborate with each other such that they can obtain a more accurate estimation for the tail risk of a certain type of insurance claims. They are willing to share some estimation results provided that other companies cannot infer client level data from the shared results. Given the sensitivity of the data, insurance companies would like to share as few results as possible. The less results transmitted and shared, the less likely that client level data can be recovered. In the proposed \dfno, we require that each machine transmit \textit{five} results to the central machine. We nevertheless consider other alternatives when further limitations on the number of results transmitted are imposed.  We compare their performance by an extensive simulation study.

The rest of the paper is organized as follows. Section 2 presents the idea for bias correction. Section 3 proposes a DC algorithm for estimating the second order parameter, defines the \df for the extreme value index and shows the main theoretical results. Section 4 provides a simulation study to confirm that the \df exhibits  superior performance compared to the  distributed Hill estimator. We discuss some extensions of our results in Section 5. The  proofs are given in the Appendix.

Throughout the paper, $a(t)\asymp b(t)$  means that both $|a(t)/b(t)|$ and $|b(t)/a(t)|$ are $O(1)$ as $t\to\infty$.

\section{Bias Correction Methodology}\label{Section:condition}
To  obtain the asymptotic normality of the distributed Hill estimator $\hat{\gamma}_{DH,k}$,   \citet{chen2021distributed} assumes the following
 second order condition. Suppose that there exist an eventually  positive  or negative function $A$ with $\lim_{t\to \infty}A(t)=0$ and a real number $\rho\le 0$ such that
$$
\lim_{t\to \infty}\frac{\frac{U(tx)}{U(t)}-x^{\gamma}}{A(t)}=x^{\gamma}\frac{x^{\rho}-1}{\rho},
$$
for all $x>0$, which is equivalent to 
\begin{equation}\label{SOC}
\lim_{t\to \infty}\frac{\log U(tx)-\log U(t)-\gamma \log x}{A(t)}=\frac{x^{\rho}-1}{\rho}.
\end{equation}
In addition,  assume that as $N \to \infty$,
\begin{equation}\label{Condtion of m,n}
	m=m(N) \to \infty,\quad  n=n(N)\to \infty,\quad  n/\log m \to \infty,
\end{equation}
and $k$ is either a fixed integer or an intermediate sequence, i.e.  $k=k(N)\to \infty, k/n \to 0$.  
Under conditions \eqref{SOC} and \eqref{Condtion of m,n}, \citet{chen2021distributed} shows that 
the distributed Hill estimator possesses the following asymptotic expansion:
$$
	\hat{\gamma}_{DH,k}-\gamma = \frac{\gamma P_{N}}{\sqrt{km}} + \frac{A(n/k)}{1-\rho}g(k,n,\rho)+\frac{1}{\sqrt{km}}o_P(1),
$$
where $P_{N} \sim N(0,1)$  and  
\begin{equation}\label{Def of g(k,n,rho)}
g(k,n,\rho) := \suit{\frac{k}{n}}^{\rho}\frac{\Gamma(n+1)\Gamma(k-\rho+1)}{\Gamma(n-\rho+1)\Gamma(k+1)},
\end{equation}
with $\Gamma$ denoting the gamma function.  By Lemma \ref{Lemma for Expectation} (see below), we have that, if $k$ is a fixed integer, then $g(k,n,\rho)\to k^{\rho}\Gamma(k-\rho+1)/\Gamma(k+1)$, as $N \to \infty$. If $k$ is an intermediate sequence, then $g(k,n,\rho)\to 1$, as $N\to \infty$.

 Since the bias term of the distributed Hill estimator is an explicit function $(1-\rho)^{-1}A(n/k)g(k,n,\rho)$, we shall estimate the bias, subtract it from the original distributed Hill estimator, which leads to the  \dfno.

The estimation of the bias term requires estimating the second order parameter $\rho$ and the second order scale function $A$ in condition \eqref{SOC}. For simplicity, we follow the bias correction literature to  assume that $\rho<0$, see e.g. \citet{de2016adapting} and \citet{ivette2007simple}.
  In order to obtain the asymptotic behavior of the estimator for $\rho$, a third order condition is often assumed.  We invoke the third order condition in \citet{alves2003new} as follows.
   Suppose that there exist an eventually positive or negative function $B$ with $\lim_{t\to \infty} B(t)=0$  and a real number $\tilde{\rho}\le 0$ such that 
  \begin{equation}\label{third order condition}
	  \begin{aligned}
 \lim_{t\to \infty} \frac{1}{B(t)}\set{\frac{\log U(tx)-\log U(t)-\gamma \log x}{A(t)}-\frac{x^{\rho}-1}{\rho}} 
   =\frac{1}{\tilde{\rho}}\suit{\frac{x^{\rho+\tilde{\rho}}-1}{\rho+\tilde{\rho}}-\frac{x^{\rho}-1}{\rho}}.
	  \end{aligned}
  \end{equation}

  Lastly, following  \citet{cai2012bias} and \citet{de2016adapting}, we use  a higher intermediate sequence $k_{\rho}$ for estimating the second order parameter $\rho$. Assume that as $N\to \infty$, $k_{\rho} = k_{\rho}(N)\to \infty, k_{\rho}/n \to 0$, and 
\begin{equation}\label{k_rho_condition}
	 \sqrt{k_{\rho}m}A(n/k_{\rho}) \to \infty, \sqrt{k_{\rho}m}A^2(n/k_{\rho})\to \lambda_1\in \mathbb{R}, \sqrt{k_{\rho}m}A(n/k_{\rho})B(n/k_{\rho})\to \lambda_2 \in \mathbb{R}.
\end{equation}
Similar to \citet{de2016adapting}, in the eventual  \df for the extreme value index, one can choose a higher number of exceedance ratios than that used in the distributed Hill estimator.
In our context, we choose a  sequence  $k_n$ such that, as $N\to \infty$, $ k_n/k_{\rho} \to 0$ and
\begin{equation}\label{k_gamma_condition}
	\sqrt{k_{n}m}A(n/k_{n}) \to \infty, \sqrt{k_{n}m}A^2(n/k_{n})\to 0, \sqrt{k_{n}m}A(n/k_{n})B(n/k_{n})\to 0.
\end{equation}
Here, similar to  the  distributed Hill estimator,  $k_n$ can be either a fixed integer or an intermediate sequence.

\section{ Main results}

We first introduce the estimator for the second order parameter $\rho$ in the distributed inference setup and study its asymptotic behavior. Then we define the \df for the extreme value index and show its asymptotic behavior. 
\subsection{Estimating the second order parameter}
If the oracle sample can be used, then there are several estimators for the second order parameter $\rho$, see e.g. \citet{alves2003new} and \citet{gomes2002semi}. However, since we cannot apply a statistical procedure to the oracle sample,
 we need to develop a DC algorithm for estimating $\rho$.
Consider the following statistics  computed based on observations on machine $j$,
$$
R_{j,k}^{(\alpha)}:=\frac{1}{k}\sum_{i=1}^{k} \set{\log M_j^{(i)}-\log M_j^{(k+1)}}^{\alpha}, \quad \alpha=1,2,3.
$$
 We request that each machine  sends the values $R_{j,k}^{(\alpha)},\alpha=1,2,3$ to the central machine. On  the central machine, we  take the  average of  the $R_{j,k}^{(\alpha)}$ statistics  to obtain 
$$
	R_k^{(\alpha)}=\frac{1}{m}\sum_{j=1}^m R_{j,k}^{(\alpha)},\quad \alpha=1,2,3.
$$
Motivated by  \citet{alves2003new}, we define the estimator for the second order parameter $\rho$ as
\begin{equation}\label{rho-estimator}
	\widehat{\rho}_{k,\tau}:=-3 \abs{\frac{T_{k,\tau}-1}{T_{k,\tau}-3}},
	\end{equation}
where 
$$
T_{k,\tau}:=
\frac{\suit{R_k^{(1)}}^{\tau}-\suit{R_k^{(2)}/2}^{\tau/2}}{\suit{R_k^{(2)}/2}^{\tau/2}-\suit{R_k^{(3)}/6}^{\tau/3}},
$$
and $\tau \ge 0$ is a tuning parameter. For $\tau=0$,  $T_{k,\tau}$ is defined by continuity. In practice, it is suggested   to choose  $\tau\in [0,1]$, see e.g. \citet{ivette2007simple} and \citet{ivette2008tail}. 

Before studying the asymptotics of $\hat{\rho}_{k,\tau}$, we first establish that  for $R^{(\alpha)}_k$ in the following proposition. Note that in this proposition, we use a general sequence $k$. Nevertheless, the proposition will be applied both for $k=k_n$ and $k=k_{\rho}$, see Section 3.2. 

\begin{prop}\label{theorem for expansion}
	Assume that the distribution function $F$ satisfies the third order condition 
\eqref{third order condition} with parameters $\gamma>0,\rho<0$ and $\tilde{\rho}\le 0$, and condition  \eqref{Condtion of m,n} holds. In addition,  suppose that an intermediate sequence $k$ satisfies that as $N \to \infty$, $ k/n \to 0$ and $\sqrt{km}A(n/k)B(n/k)=O(1),\sqrt{km}A^2(n/k)=O(1)$. Then  for suitable versions of the functions $A$ and $B$, denoted as $A_0$ and $B_0$ (see Lemma \ref{lemma for third order equality} below), we have that as $N \to \infty$,
\begin{itemize}
	\item[(i)]
	$$
	\begin{aligned}
		\sqrt{km}\suit{R_k^{(1)}-\gamma} -\gamma P_N^{(1)}-\frac{g(k,n,\rho)}{1-\rho}\sqrt{km}A_0(n/k)
			-\frac{g(k,n,\rho+\tilde{\rho})}{1-\rho-\tilde{\rho}}\sqrt{km}A_0(n/k)B_0(n/k)=o_p(1),
			\end{aligned}
	$$ 
	\item[(ii)]
	$$
		\begin{aligned}
		&\sqrt{km}\suit{R_k^{(2)}-2\gamma^2} -\gamma^2 P_N^{(2)}-2\gamma\sqrt{km}A_0(n/k)\frac{g(k,n,\rho)}{\rho}\set{\frac{1}{(1-\rho)^2}-1} \\
		&\quad \quad  -\sqrt{km}A_0^2(n/k)\frac{g(k,n,2\rho)}{\rho^2}\suit{\frac{1}{1-2\rho}-\frac{2}{1-\rho}+1}\\
		&\quad \quad -2\gamma \sqrt{km}A_0(n/k)B_0(n/k)\frac{g(k,n,\rho+\tilde{\rho})}{\rho+\tilde{\rho}}\set{\frac{1}{(1-\rho-\tilde{\rho})^2}-1}=o_p(1),
		\end{aligned}
	$$
	\item [(iii)]
	$$
		\begin{aligned}
	&\sqrt{km}\suit{R_k^{(3)}-6\gamma^3} -\gamma^3 P_N^{(3)}-6\gamma^2\sqrt{km}A_0(n/k)\frac{g(k,n,\rho)}{\rho}\set{\frac{1}{(1-\rho)^3}-1} \\
	&\quad \quad  -3\gamma\sqrt{km}A_0^2(n/k)\frac{g(k,n,2\rho)}{\rho^2}\set{\frac{1}{(1-2\rho)^2}-\frac{2}{(1-\rho)^2}+1}\\
	&\quad \quad -6\gamma^2 \sqrt{km}A_0(n/k)B_0(n/k)\frac{g(k,n,\rho+\tilde{\rho})}{\rho+\tilde{\rho}}\set{\frac{1}{(1-\rho-\tilde{\rho})^3}-1}=o_P(1),
	\end{aligned}
	$$
\end{itemize}
where $(P_N^{(1)}, P_N^{(2)}, P_N^{(3)})^T\sim N(\bf{0},\bSigma)$ with
$$
\bSigma = \suit{
	\begin{array}{lll}
		1 & 4 & 18 \\
		4 & 20 & 98 \\
		18 & 98 & 684 
	\end{array}
}.
$$

\end{prop}

Applying Proposition \ref{theorem for expansion} leads to 
 the asymptotic behavior of $\hat{\rho}_{k,\tau}$ as follows.

\begin{theorem}\label{Theorem of rho}
	Assume that the distribution function $F$ satisfies the third order condition 
\eqref{third order condition} with parameters $\gamma>0,\rho<0$ and $\tilde{\rho}\le 0$, and condition   \eqref{Condtion of m,n} holds.
Suppose that the intermediate sequence $k_{\rho}$ satisfies  condition \eqref{k_rho_condition}.  Then as  $N\to \infty$, for each $\tau\ge 0$,
	$$
	\sqrt{k_{\rho}m}A_0(n/k_{\rho})(\hat{\rho}_{k_{\rho},\tau}-\rho) =O_P(1),
	$$
	where $\hat{\rho}_{k_{\rho},\tau}$ is defined in \eqref{rho-estimator}.
\end{theorem}

\subsection{\ldf for the extreme value index}
 Motived by \citet{de2016adapting},  we define the following estimator as the \df for  the extreme value index:
\begin{equation}\label{gamma_estimator}
	\tilde{\gamma}_{k_n,k_{\rho},\tau}:=R_{k_{n}}^{(1)}-\frac{R_{k_{n}}^{(2)}-2\suit{R_{k_{n}}^{(1)}}^2}{2R_{k_{n}}^{(1)}\hat{\rho}_{k_{\rho},\tau} (1-\hat{\rho}_{k_{\rho},\tau})^{-1}},
\end{equation}
where $\tau\ge 0$ is a tuning parameter. Notice that the estimator $\tilde{\gamma}_{k_n,k_{\rho},\tau}$ in \eqref{gamma_estimator} adheres to a DC algorithm since each machine only  sends  five values  $\set{R_{j,k_n}^{(1)}, R_{j,k_n}^{(2)}, R_{j,k_{\rho}}^{(1)}, R_{j,k_{\rho}}^{(2)}, R_{j,k_{\rho}}^{(3)}}$ to the central machine.
\begin{remark}
The statistic $R_{k_n}^{(1)}$ is the original distributed Hill estimator $\hat{\gamma}_{DH,k_n}$.
\end{remark}
The following theorem  shows  the  asymptotic  normality  of the \dfno.
\begin{theorem}\label{Theorem : gamma}
	Assume that the distribution function $F$ satisfies the third order condition 
\eqref{third order condition} with parameters $\gamma>0,\rho<0$ and $\tilde{\rho}\le 0$, and  condition \eqref{Condtion of m,n} holds.
	Suppose that $k_{\rho}, k_n$ satisfy  conditions \eqref{k_rho_condition} and \eqref{k_gamma_condition} respectively. Then as $N \to \infty$, for each $\tau\ge 0$,
	$$
\sqrt{k_nm}\suit{\tilde{\gamma}_{k_n,k_{\rho},\tau}-\gamma}\stackrel{d}{\to} N\bra{0,\gamma^2 \set{1+\suit{\rho^{-1}-1}^2}}.
	$$
\end{theorem}

\begin{remark}
We investigate the conditions in Theorem \ref{Theorem : gamma} to determine the range of $m$ (and $k$) such that the oracle property holds.
The last statement in Condition \eqref{Condtion of m,n}, $n/\log m \to \infty$ as $N\to\infty$, provides an upper bound for $m$ as $m=o(N/\log N)$ as $N\to\infty$. Condition \eqref{k_gamma_condition} leads to an upper bound for $k_nm$: based on the second order condition \eqref{SOC}, we need to have $k_nm = O(N^{\xi})$ with $\xi <1$.  Clearly, for the number of machine $m$, the second upper bound is a stricter requirement than the first.
\end{remark}

\begin{remark}\label{remark: DC property}
	The limit distribution in Theorem \ref{Theorem : gamma} is the same as that of the  bias corrected Hill estimator based on the oracle sample, see for example \citet{de2016adapting}. In other words, the \df  achieves the oracle property regardless whether $k_n$ is a fixed integer or an intermediate sequence. 
	\citet{chen2021distributed} shows that  when $k_n$ is a fixed integer, the distributed Hill estimator may possess a higher  bias than that of the oracle Hill estimator. Consequently, the distributed Hill estimator achieves the oracle property only if additional conditions are assumed, see Corollary 1 therein. If the additional conditions fail, the violation of the oracle property is due to the difference in the asymptotic biases of the two estimators. By contrast, the \df achieves the oracle property without any additional assumption when $k_n$ is a fixed integer. This is due to the fact that the asymptotic bias was corrected.

	Nevertheless, if Condition \eqref{k_gamma_condition} is violated in the following sense:
as $N\to\infty$,  $\sqrt{k_nm}A^2(n/k_n) \to \lambda_3$ and $\sqrt{k_nm}A(n/k_n)B(n/k_n) \to\lambda_4 $ where $\lambda_3\neq 0$ or $\lambda_4\neq 0$, then the oracle bias corrected estimator will possess a non-zero asymptotic bias. In this case, the \df may not possess the oracle property.
\end{remark}

\begin{remark} We investigate the optimal choice for $k_n$ in terms of the level of the asymptotic root mean squared error (RMSE). We first consider the \dfno. To simplify the discussion, we focus on the case $A(t) \asymp t^{\rho}, B(t)\asymp t^{\tilde{\rho}}$  as $t\to\infty$. The best attainable
	rate of convergence is achieved when squared bias and variance are of the same order, that is, 	when 
	$$
	   \frac{1}{\sqrt{k_nm}} \asymp A(n/k_n)\set{A(n/k_n)+B(n/k_n)},
	$$
	as $N\to\infty$.
	Solving $k_n$ yields that $k_n^{DC} \asymp N^{-2\rho^{*}/(1-2\rho^{*})} m^{-1}$ as $N\to\infty$,  where $\rho^*=\rho+\max(\rho,\tilde{\rho})$.  

	Similarly, we obtain the optimal choice of $k_n$ in a single machine as  $k_n^{Single}\asymp  n^{-2\rho^{*}/(1-2\rho^{*})} $. Note that, as $N\to\infty$, $k_n^{DC}/k_n^{Single} \asymp m^{-1/(1-2\rho^*)} \to 0$. We conclude that the two optimal choices do not match each other: the optimal choice of $k_n$ at each individual machine is too high for optimally using the \dfno. In practice, for example, in the insurance claim example, to make use of the \dfno, one needs to coordinate the choice of $k_n$ at all insurance companies instead of allowing each insurance company to choose the optimal level of $k_n$ based on their own data. 
\end{remark}

\section{Simulation Study}

\subsection{Comparison with the original distributed Hill estimator}
In this subsection, we conduct a simulation study to demonstrate the finite sample performance of the  \df for the extreme value index. Data are simulated from three distributions: the  Fr\'echet distribution, $F(x)= \exp\suit{-x^{-1}}, x>0$; the Burr distribution, $F(x)=1-(1+x^{1/2})^{-2}, x>0$; and  the absolute Cauchy distribution with the density function $f(x)=2/\set{\pi(1+x^2)}, x>0$. The first, second  and third order indices of the three distributions are listed in Table \ref{Distributions for simulation}.  We generate $r=1000$ samples with sample size $N=10000$.   The value of  $k_{\rho}$ is chosen to be $[n^{0.98}]$ as suggested by \citet{cai2012bias}, where $[x]$ denotes the largest integer less than or equal to $x$.

\begin{table}[htbp]
	\centering
	\begin{tabular}{cccc}
		\hline
			 & Fr\'echet & Burr & Absolute Cauchy \\
		\hline
		$\gamma$ & $1$& $1$& $1$\\
		$\rho$ & $-1$& $-1/2$ & $-2$ \\
		$\tilde{\rho}$ &$-1$ &$-1/2$ &$-4$ \\
	\hline
	\end{tabular}
	\caption{The first, second and third order indicies for the distributions.}
	\label{Distributions for simulation}
\end{table}

To apply the \dfno, we use the following procedure:
\begin{itemize}
	\item[1.] On each machine $j$, we calculate $R_{j,k_n}^{(1)}$, $R_{j,k_n}^{(2)}$, $R_{j,k_{\rho}}^{(1)}$, $R_{j,k_{\rho}}^{(2)}$, $R_{j,k_{\rho}}^{(3)}$ and transmit them to the central machine. 
	\item[2.] On the central machine, we take the average of the $R_{j,k_n}^{(1)}$, $R_{j,k_n}^{(2)}$, $R_{j,k_{\rho}}^{(1)}$, $R_{j,k_{\rho}}^{(2)}$, $R_{j,k_{\rho}}^{(3)}$ statistics collected from the $m$ machines  to obtain $R_{k_n}^{(1)}$, $R_{k_n}^{(2)}$, $R_{k_{\rho}}^{(1)}$, $R_{k_{\rho}}^{(2)}$, $R_{k_{\rho}}^{(3)}$.
	\item[3.] On the central machine, we estimate the second order parameter $\rho$ by  \eqref{rho-estimator} with $k=k_{\rho}$.  The value of the tuning parameter $\tau$ is set at $0,0.5$ and $1$.
	\item[4.] On the central machine, we estimate the extreme value index by \eqref{gamma_estimator} for various values of $k_n$, using $\hat{\rho}_{k_{\rho},\tau}$.
\end{itemize}

We assume that the  $N=10000$ observations are stored in $m=1,20,100$ machines with $n=N/m$ observations each.  Note that the case $m=1$ corresponds to applying the  statistical procedure to the oracle sample directly. The corresponding estimator is therefore the oracle estimator.

Figure \ref{figure: Absolute Bias.} shows the absolute bias against various levels of $k_n$ for the three distributions with $m=20$. The results for other values of $m$ show similar patterns and are thus omitted.
 We observe that,   the \df $\tilde{\gamma}_{k_n,k_{\rho},\tau}$ generally has superior performance compared to the original distributed Hill estimator $\hat{\gamma}_{DH,k_n}$. As $k_n$ increases, the bias of the distributed Hill estimator increases, while the \df has almost zero bias except for very high level of $k_n$. This is in line with the asymptotic theory.
In addition, the choice of $\tau$ affects the performance of the \dfno.  When $\rho<-1$ (absolute Cauchy distribution), $\tau=1$ is a better choice than $\tau=0$. When $\rho\ge -1$ (Fr\'echet distribution and Burr distribution), $\tau=0$ is a better choice than $\tau=1$. This  is in line with the findings in \citet{alves2003new}.
 
 \begin{figure}
	\centering
\subfigcapskip=-8pt
\subfigure[Fr\'echet]{
\includegraphics[width=0.32\textwidth]{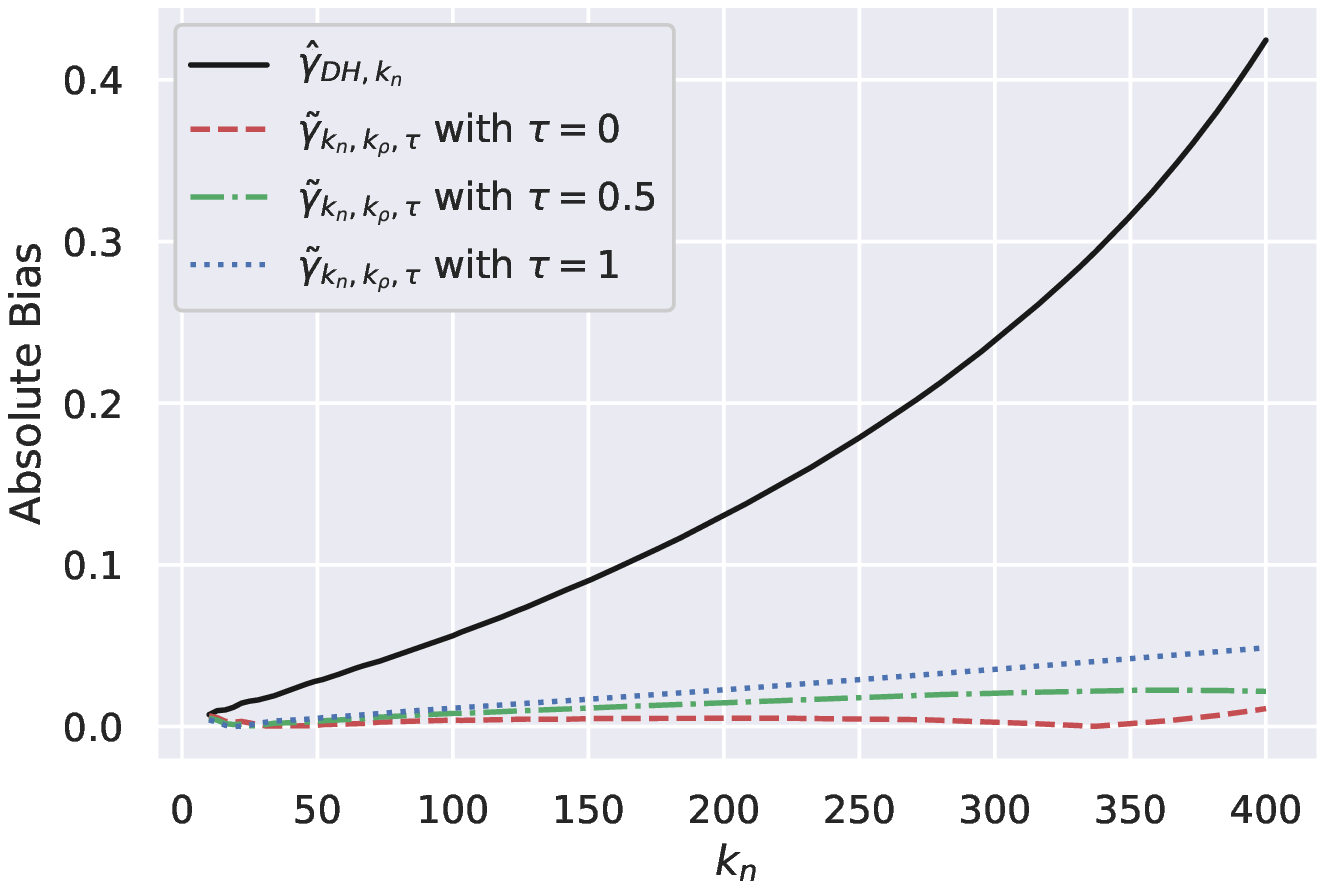}
}
\hspace{-3ex}
\subfigure[Burr]{
	\includegraphics[width=0.32\textwidth]{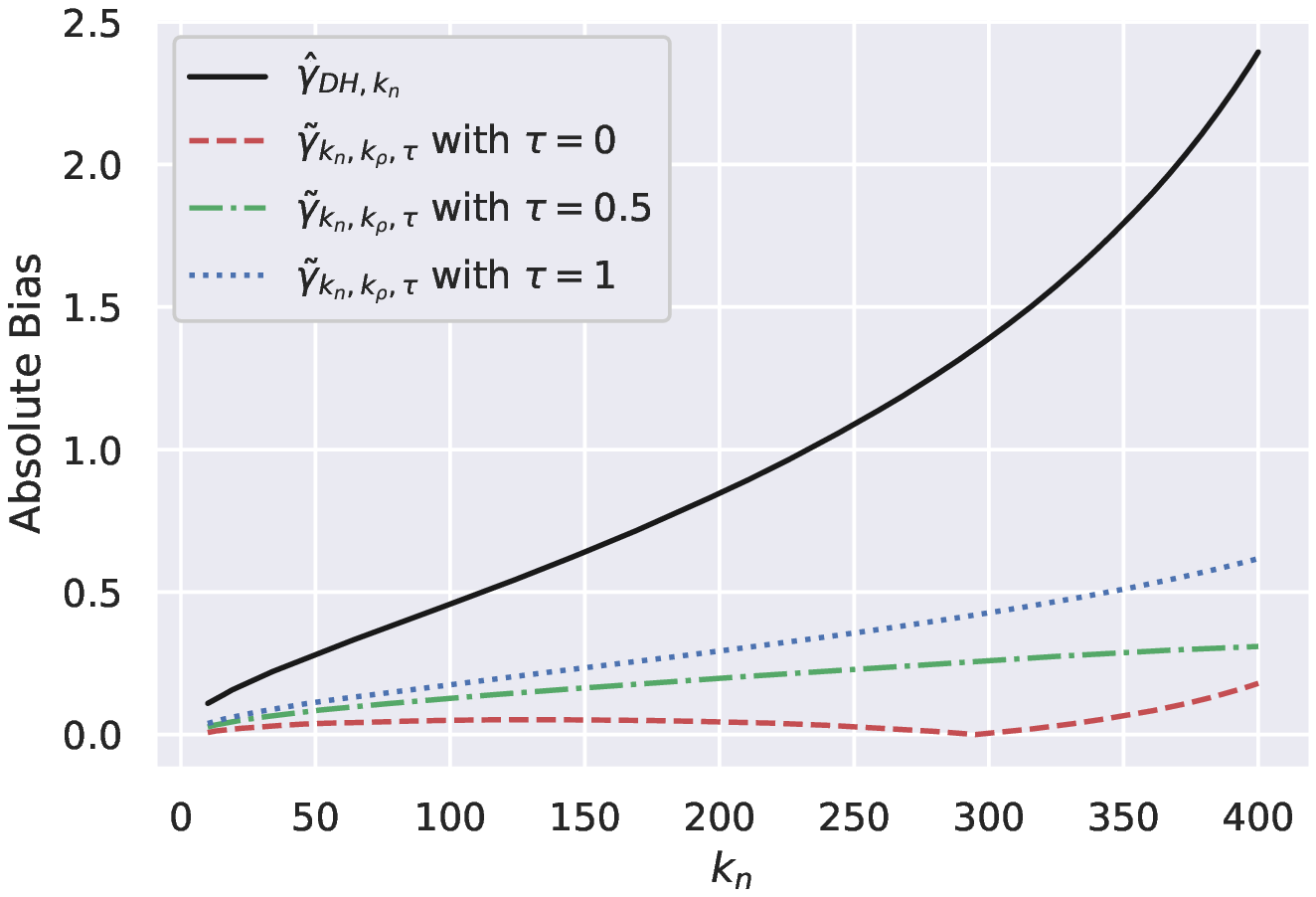}
}
\hspace{-3ex}
\subfigure[Absolute Cauchy]{
	\includegraphics[width=0.32\textwidth]{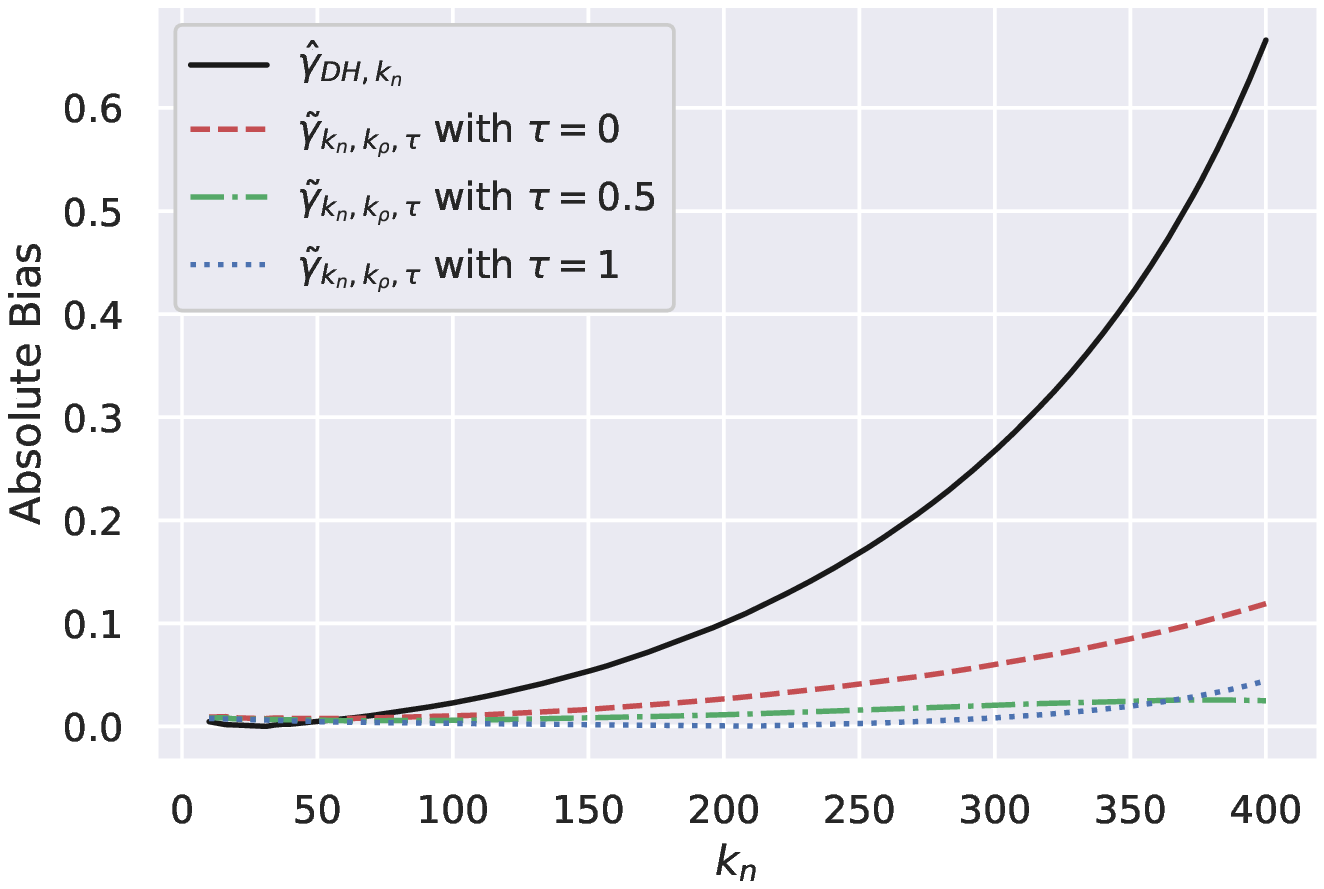}
}
\caption{Absolute bias    for different levels of $k_n$ with $m=20$.}
\label{figure: Absolute Bias.}
\end{figure}

% \begin{figure}[htbp]
% 	\centering
% 	\subfigure[m=1]{
% 	\includegraphics[width=0.33\textwidth]{program/simu1/figs/Frechet_m1_0}
% 	}
% 	\hspace{-4ex}
% 	 \subfigure[m=20]{
% 	 \includegraphics[width=0.33\textwidth]{program/simu1/figs/Frechet_m20_0}
% 	 }
% 	\hspace{-4ex}
% 	\subfigure[m=100]{
% 		\includegraphics[width=0.33\textwidth]{program/simu1/figs/Frechet_m100_0}
% 		}
% 	\caption{Absolute bias for the unit Fr\'echet distribution.}
% 	\label{figure frechet}
% \end{figure}

% \begin{figure}[htbp]
% 	\centering
% 	\subfigure[m=1]{
% 	\includegraphics[width=0.33\textwidth]{program/simu1/figs/Burr1_m1_0}
% 	}
% 	\hspace{-4ex}
% 	\subfigure[m=20]{
% 	\includegraphics[width=0.33\textwidth]{program/simu1/figs/Burr1_m20_0}
% 	}
% 	\hspace{-4ex}
% 	\subfigure[m=100]{
% 		\includegraphics[width=0.33\textwidth]{program/simu1/figs/Burr1_m100_0}
% 		}
% 	\caption{Absolute bias for the Burr distribution.}
% 	\label{figure Burr1}
% \end{figure}

% \begin{figure}[htbp]
% 	\centering
% 	\subfigcapskip=-8pt
% 	\subfigure[m=1]{
% 	\includegraphics[width=0.33\textwidth]{program/simu1/figs/Cauchy_m1_0}
% 	}

% 	\hspace{-4ex}
% 	\subfigure[m=20]{
% 	\includegraphics[width=0.33\textwidth]{program/simu1/figs/Cauchy_m20_0}
% 	}
% 	\hspace{-4ex}
% 	\subfigure[m=100]{
% 		\includegraphics[width=0.33\textwidth]{program/simu1/figs/Cauchy_m100_0}
% 		}
% 	\caption{Absolute bias for the absolute Cauchy distribution.}
% 	\label{figure Cauchy}
% \end{figure}

Next, we compare the performance of the \df for different values of $m$. In this comparison, we fix 
$\tau = 0.5$. We plot the RMSE of the estimators against various levels of $k_n m$   in Figure \ref{figure: MSE with m}. For the  Fr\'echet distribution and the  absolute Cauchy distribution, the performance of the \df is generally not sensitive to the variation in  $m$. The performance across different values of $m$ is comparable to the case $m=1$, i.e., the oracle property holds.  For the Burr distribution, the oracle property only holds when $k_nm$ is low. When $k_nm$ is high, the oracle bias corrected estimator fails to correct the bias and the RMSE for the distributed estimator is higher than that of the oracle estimator. This observation is in line with the theoretical discussion in Remark \ref{remark: DC property}. 

\begin{figure}
	\centering
\subfigcapskip=-8pt
\subfigure[Fr\'echet]{
\includegraphics[width=0.32\textwidth]{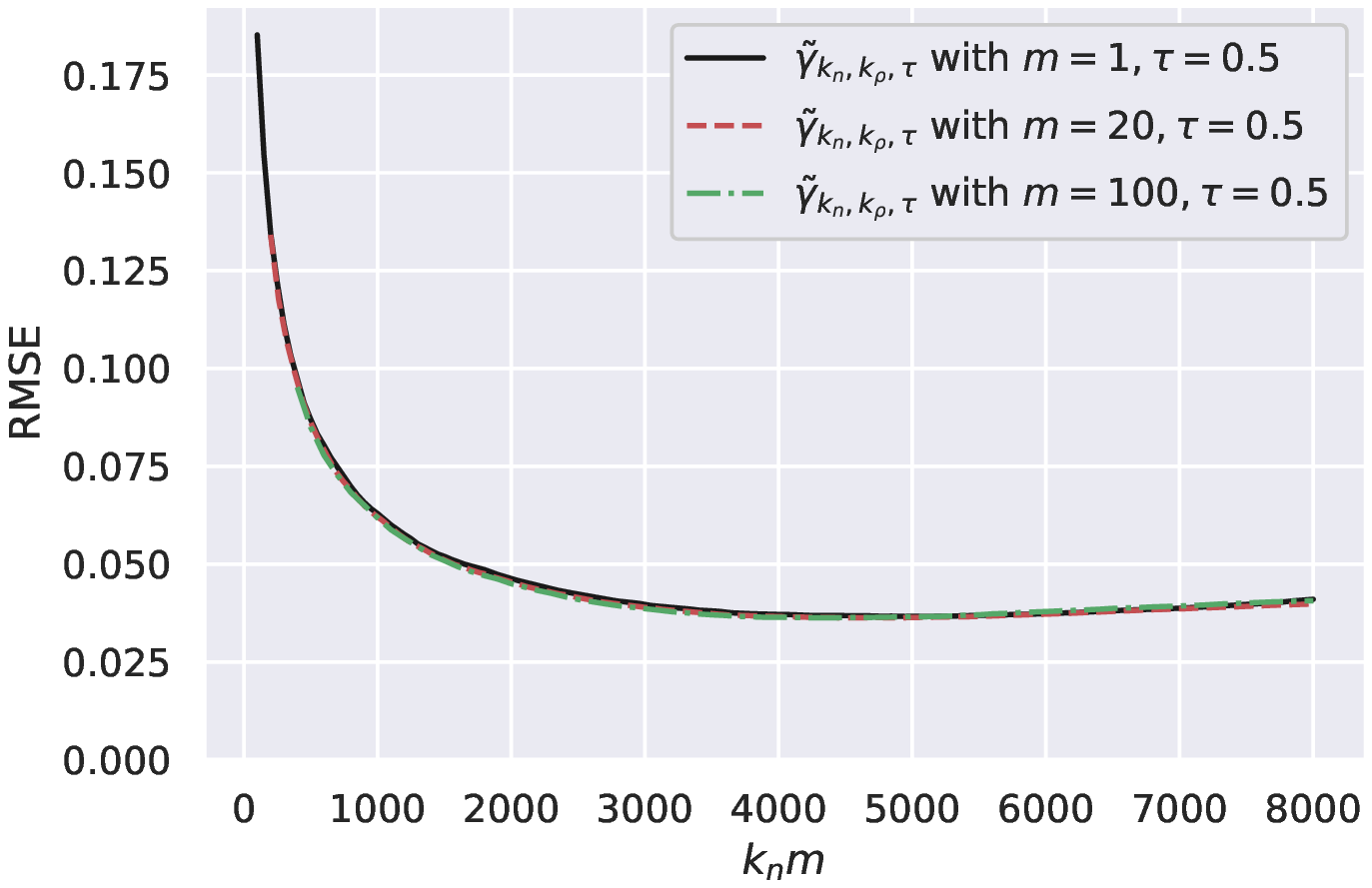}
}
\hspace{-3ex}
\subfigure[Burr]{
	\includegraphics[width=0.32\textwidth]{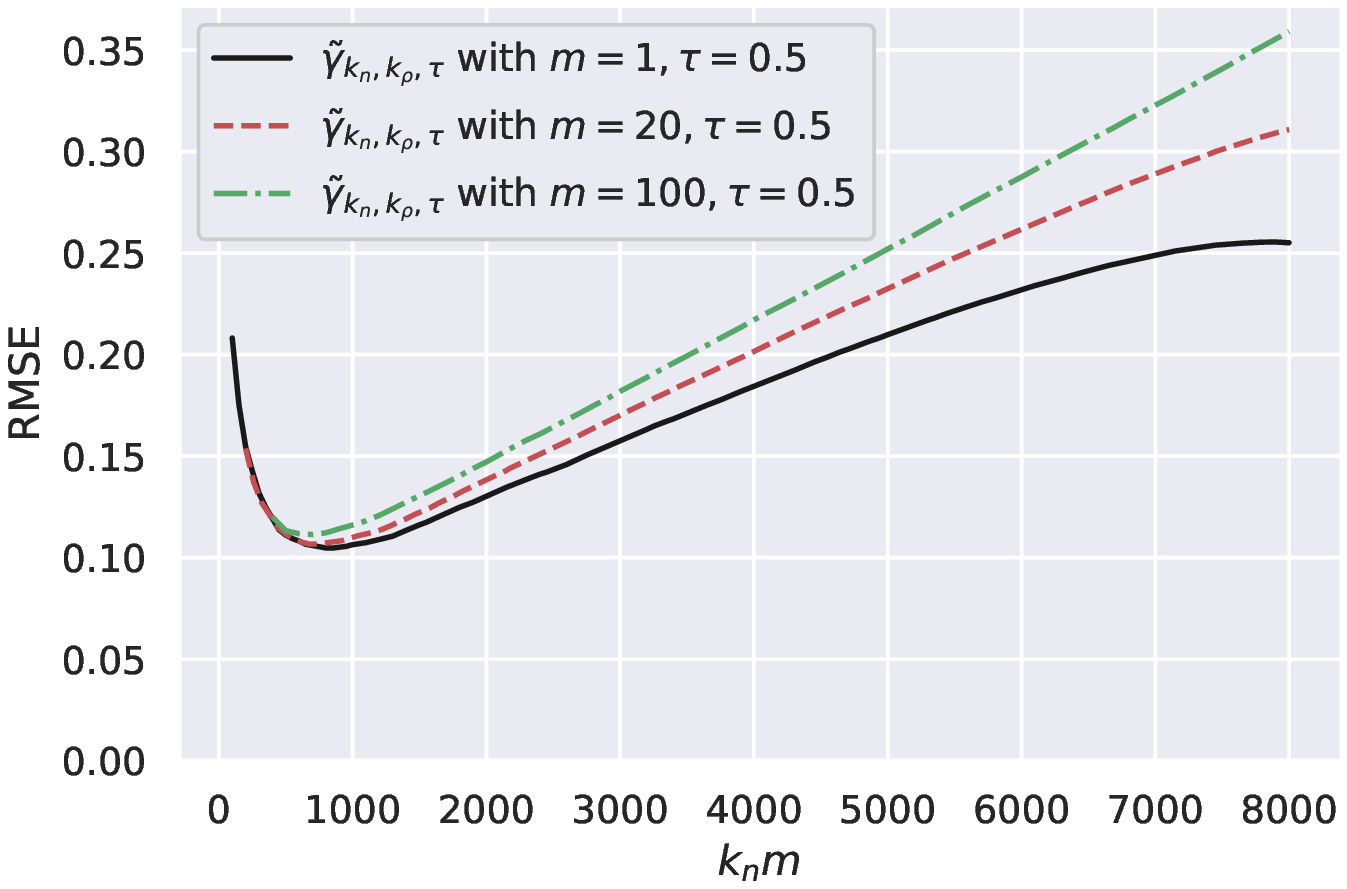}
}
\hspace{-3ex}
\subfigure[Absolute Cauchy]{
	\includegraphics[width=0.32\textwidth]{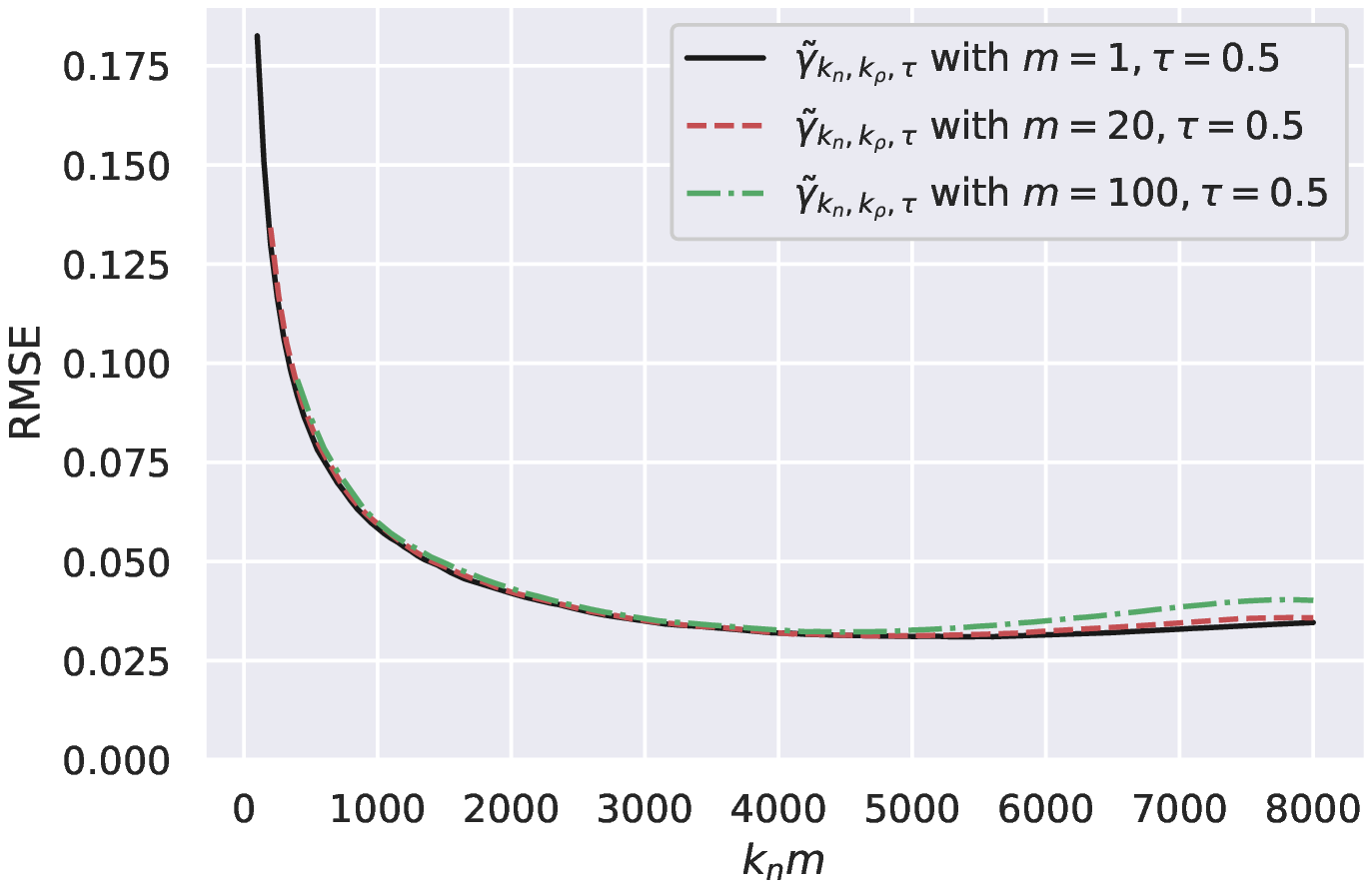}
}
\caption{RMSE for  different levels of $k_nm$.}
\label{figure: MSE with m}
\end{figure}

% \begin{figure}[htbp]
% 	\centering
% 	\subfigure[Unit Fr\'echet]{
% 	\includegraphics[width=0.35\textwidth]{program/simu1/figs/Frechet_m20_1}
% 	}
% 	\hspace{-5ex}
% 	\subfigure[Burr]{
% 	\includegraphics[width=0.33\textwidth]{program/simu1/figs/Burr1_m20_1}
% 	}
% 	\hspace{-5ex}
% 	\subfigure[Absolute Cauchy]{
% 		\includegraphics[width=0.33\textwidth]{program/simu1/figs/Cauchy_m20_1}
% 		}
% 	\caption{MSE for the unit Fr\'echet, Burr and absolute Cauchy distribution for $m=20$. }
% 	\label{figure: figure MSE}
% \end{figure}

 One important advantage of bias correction method in extreme value statistics is that the bias corrected estimator is relatively insensitive to the number of tail observations used in estimation, when applying it to a single sample. This advantage might be less pronounced for the distributed estimator since increasing $k_n$ by 1 will effectively lead to an increase of the number of tail observations by $m$. To examine this effect, we compare the single sample performance of the \df with different values of $m$. Figure \ref{Single sample} shows the plot of the estimates against various levels of $k_nm$ based on one single sample consisting of $10000$ observations. We observe that the path of the \df across different values of $m$ is comparable to the case $m=1$. In other words, the \df inherits the advantage of the bias correction estimator: it stabilizes the performance over a broader range of $k_n m$.
\begin{figure}[htbp]
	\centering
\subfigcapskip=-8pt
\subfigure[Fr\'echet]{
\includegraphics[width=0.32\textwidth]{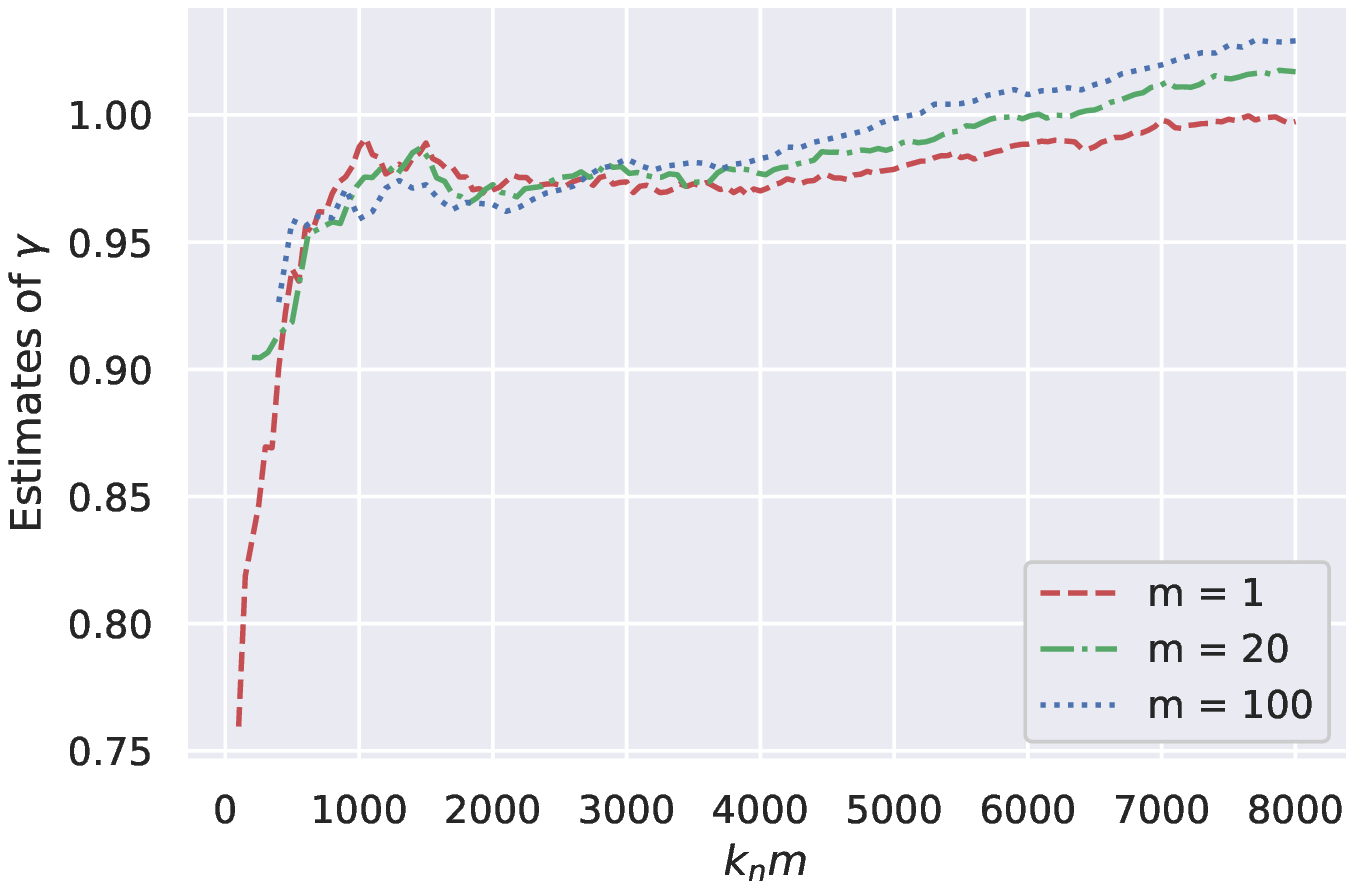}
}
\hspace{-3ex}
\subfigure[Burr]{
	\includegraphics[width=0.32\textwidth]{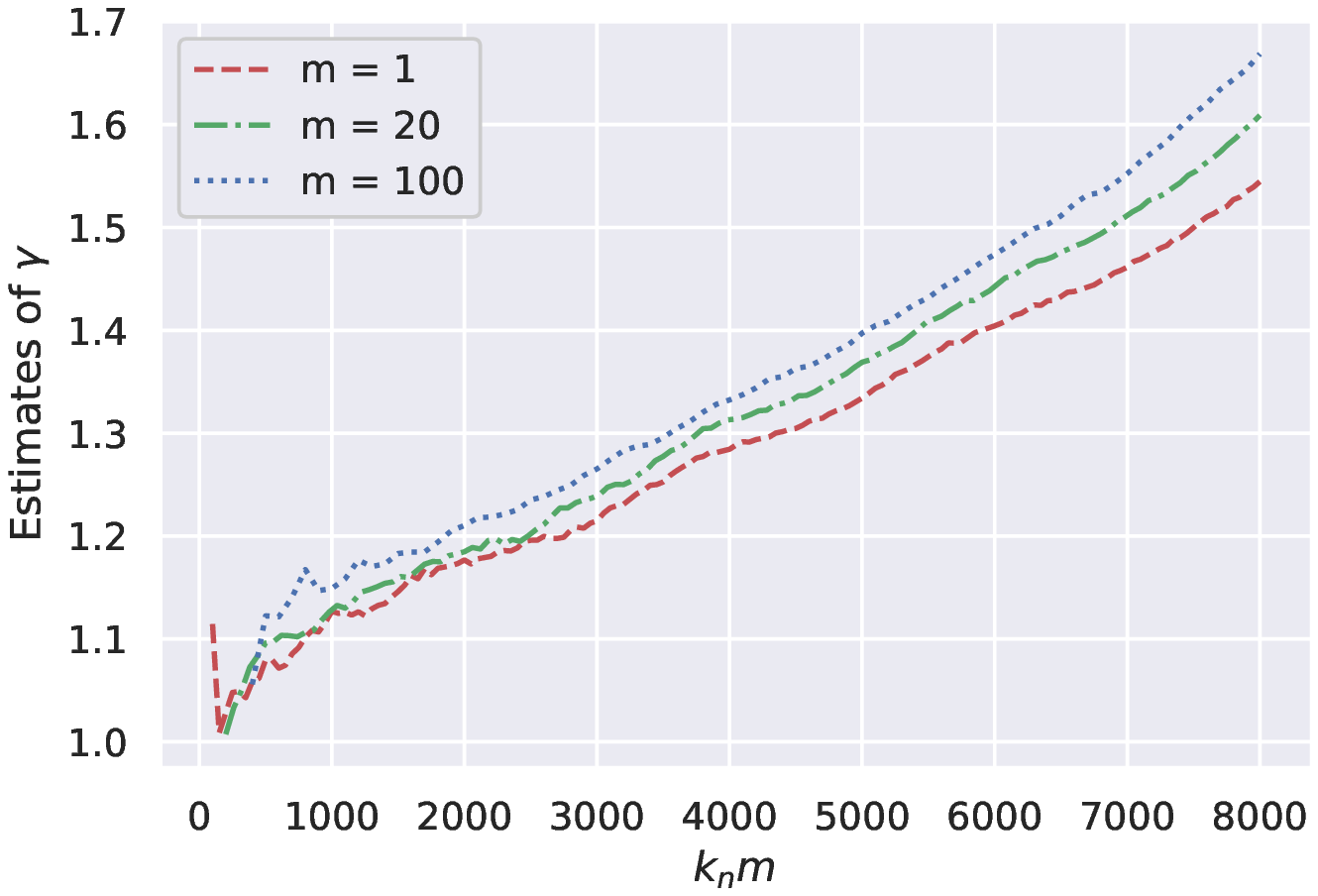}
}
\hspace{-3ex}
\subfigure[Absolute Cauchy]{
	\includegraphics[width=0.32\textwidth]{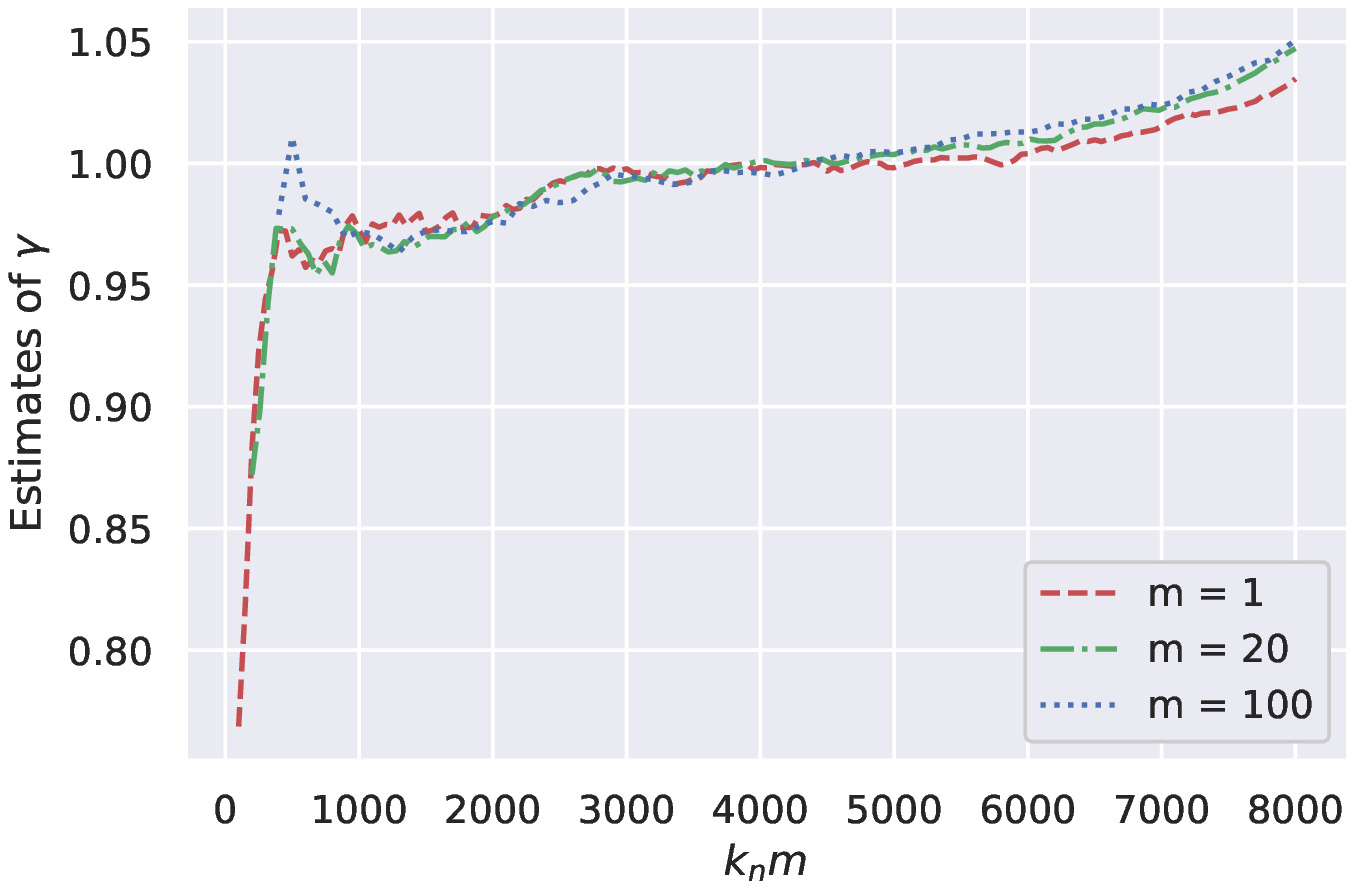}
}
\caption{Single sample performance. }
\label{Single sample}
\end{figure}

Finally, we examine the impact of choosing $k_{\rho}$. In this comparison, we fix $m=20$ and $\tau =0.5$, and  consider three choices of $k_{\rho} = [n^{0.96}],[n^{0.98}], [n^{0.99}]$. 
Figure \ref{figure rho} shows the plots of the RMSE against various levels of $k_n$. For the Fr\'echet and the absolute Cauchy distribution, the \df is not sensitive to the choice of $k_{\rho}$, while $k_{\rho} = [n^{0.98}]$ performing slight better for high level of $k_n$.  For the Burr distribution, $k_{\rho} = [n^{0.96}]$ yields slightly better performance. Nevertheless, the RMSEs for the three choices of $k_{\rho}$ are still comparable when $k_n$ is low.

\begin{figure}[htbp]
	\centering
\subfigcapskip=-8pt
\subfigure[Fr\'echet]{
\includegraphics[width=0.32\textwidth]{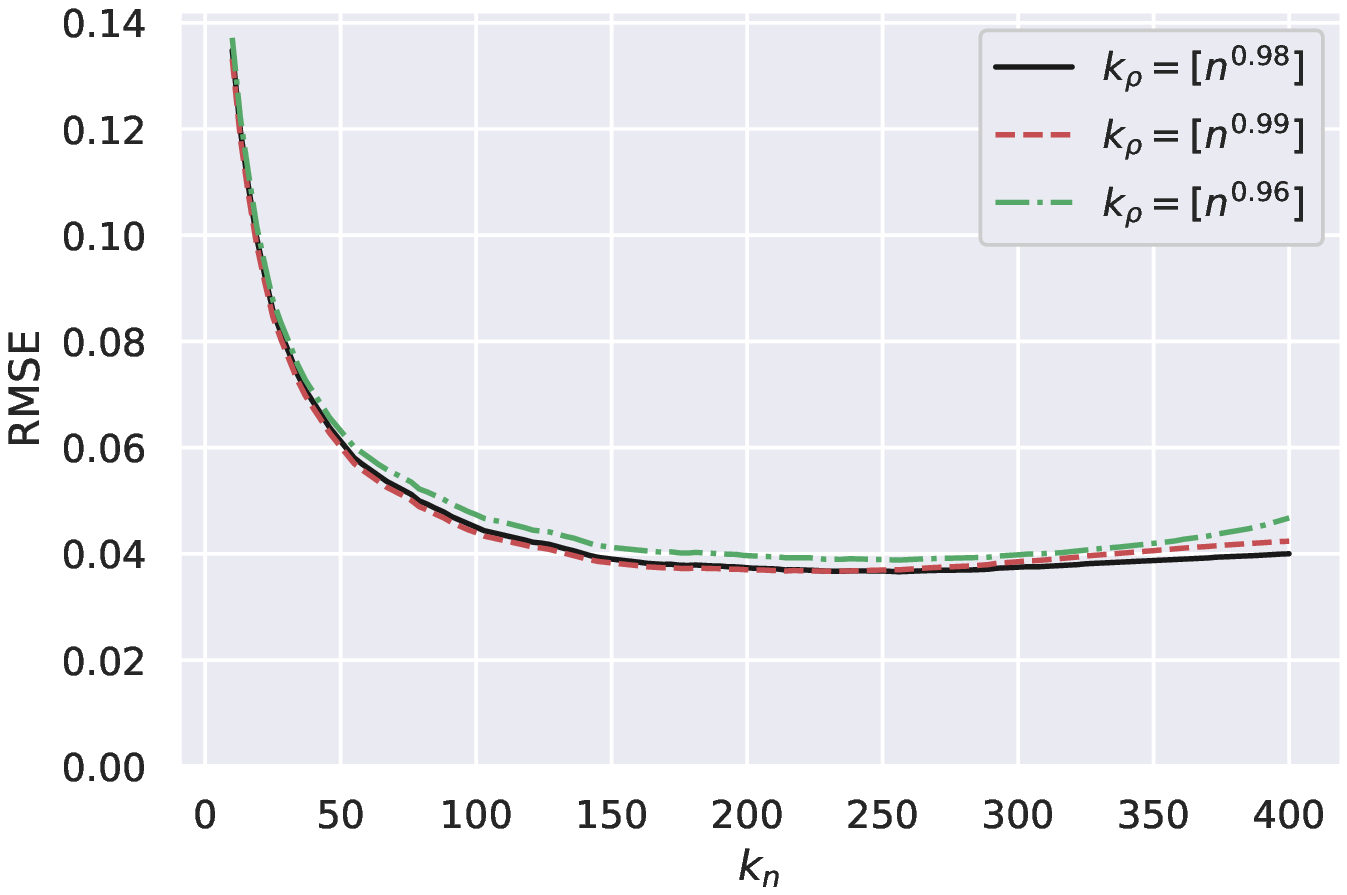}
}
\hspace{-3ex}
\subfigure[Burr]{
	\includegraphics[width=0.32\textwidth]{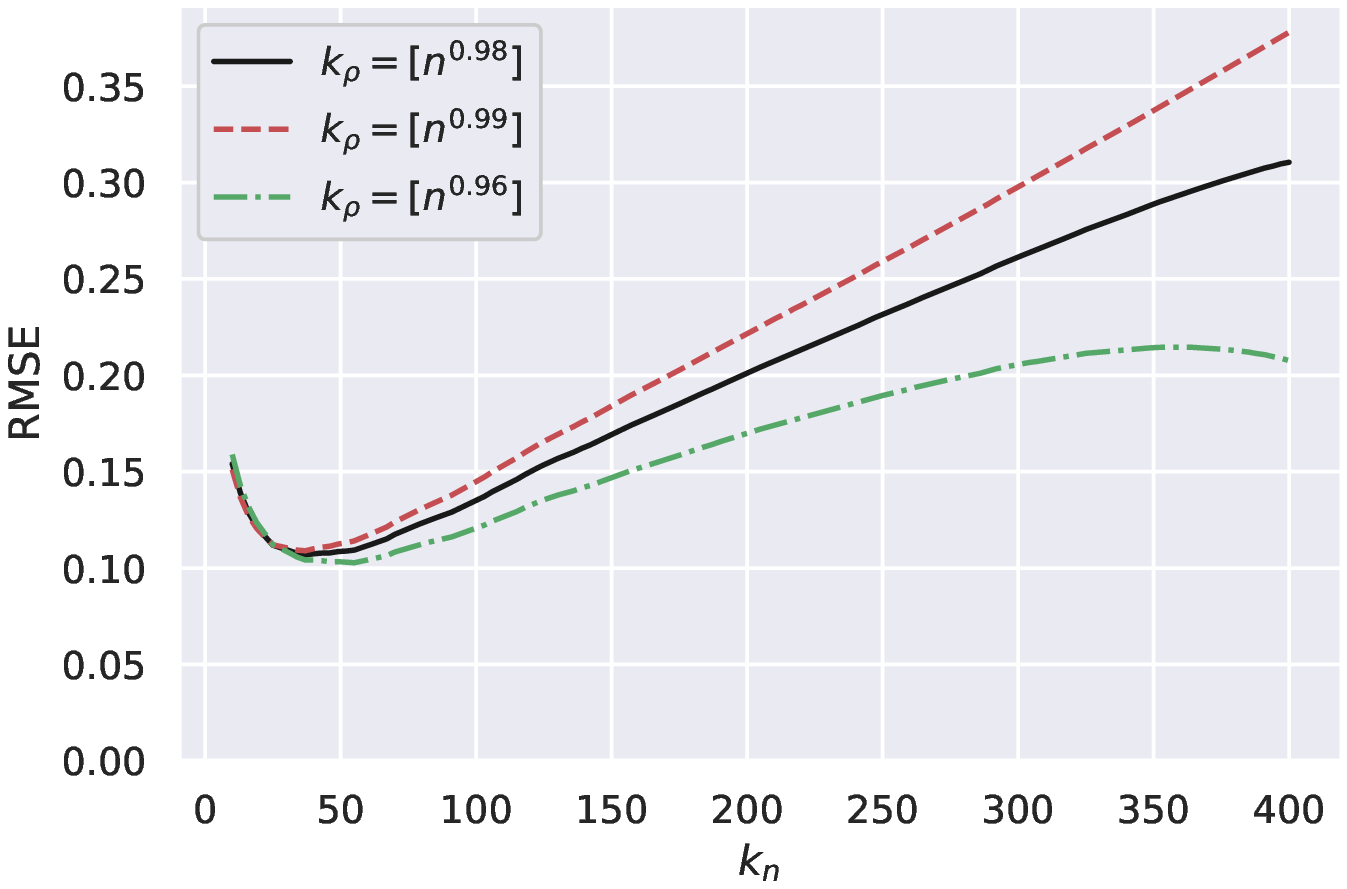}
}
\hspace{-3ex}
\subfigure[Absolute Cauchy]{
	\includegraphics[width=0.32\textwidth]{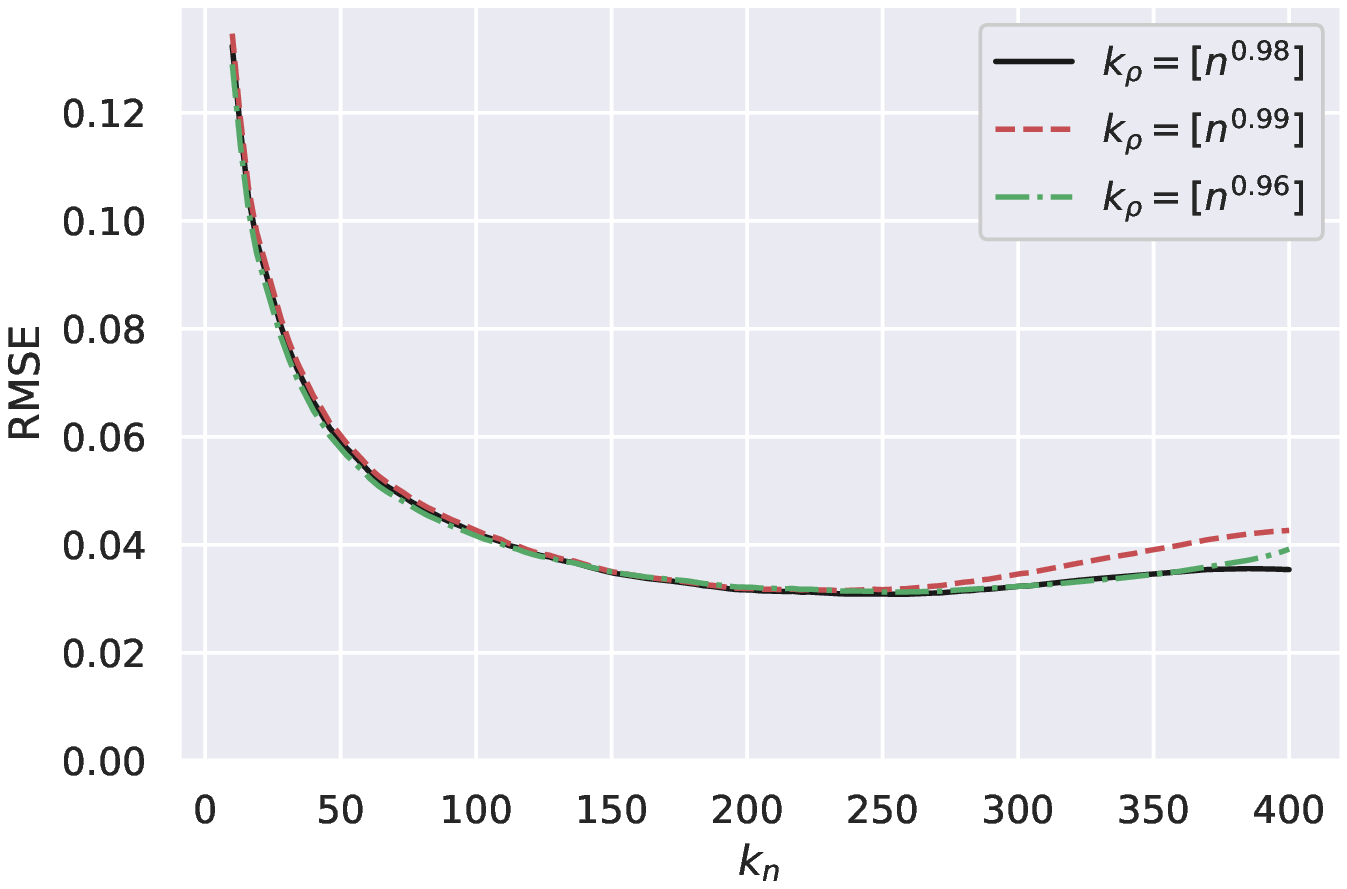}
}
\caption{Performance for different choices of $k_{\rho}$.}
\label{figure rho}
\end{figure}

\subsection{Further limitation for transmission}
Recall that for the \dfno, we need to transmit five statistics from each of the  $m$ machines to the central machine. If there are further limitations on the number of results that can be transmitted, such as only three, or even one statistic  can be transmitted,  the estimation procedure in Section 4.1 will not be applicable. In this subsection, we consider two alternative procedures for bias correction in the distributed inference setup with fewer number of transmissions.

Firstly, we consider a bias correction procedure if only three statistics  can be transmitted. We can estimate the second order parameter $\rho$ on each machine and transmit the estimates for $\rho$ to the central machine. The detailed procedures are given as follows:
\begin{itemize}
	\item On each machine $j$, we calculate  $R_{j,k_n}^{(1)}$, $R_{j,k_n}^{(2)}$, $R_{j,k_{\rho}}^{(1)}$, $R_{j,k_{\rho}}^{(2)}$, $R_{j,k_{\rho}}^{(3)}$.
	\item On each machine $j$, we estimate the second order parameter $\rho$ by 
	\begin{equation}\label{Eq: local rho estimation}
		\hat{\rho}_{j,k_{\rho},\tau}:=-3 \abs{\frac{T_{j,k_{\rho},\tau}-1}{T_{j,k_{\rho},\tau}-3}},
		\end{equation}
			with 
			$$
			T_{j,k_{\rho},\tau}:=
			\frac{\suit{R_{j,k_{\rho}}^{(1)}}^{\tau}-\suit{R_{j,k_{\rho}}^{(2)}/2}^{\tau/2}}{\suit{R_{j,k_{\rho}}^{(2)}/2}^{\tau/2}-\suit{R_{j,k_{\rho}}^{(3)}/6}^{\tau/3}},
			$$
			and   transmit  $\hat{\rho}_{j,k_{\rho},\tau}, R_{j,k_n}^{(1)}, R_{j,k_n}^{(2)}$  to the central machine. 
	\item On the central machine, we take the average of the $\hat{\rho}_{j,k_{\rho},\tau}, R_{j,k_n}^{(1)}, R_{j,k_n}^{(2)}$  to obtain
	$$
	\tilde{\rho}_{k_{\rho},\tau}=\frac{1}{m }\sum_{j=1}^m \hat{\rho}_{j,k_{\rho},\tau},\ R_{k_n}^{(1)} =\frac{1}{m }\sum_{j=1}^m R_{j,k_n}^{(1)}, \ R_{k_n}^{(2)} =\frac{1}{m }\sum_{j=1}^m R_{j,k_n}^{(2)}.
	$$
	\item On the central machine, we estimate the extreme value index by 
	$$
	\tilde{\gamma}_{k_n,k_{\rho},\tau}^{(2)}:=R^{(1)}_{k_{n}}-\frac{R_{k_n}^{(2)}-2\suit{R_{k_n}^{(1)}}^2}{2R_{k_n}^{(1)}\tilde{\rho}_{k_{\rho},\tau} (1-\tilde{\rho}_{k_{\rho},\tau})^{-1}}.	
	$$
\end{itemize}

Secondly, we consider a bias correction procedure if only one statistic can be transmitted.
We can conduct  bias correction on each machine and transmit the estimates using  the bias-corrected  Hill estimator to the central machine. Then we take the average of these estimates on the  central machine. 
In this procedure,  each machine only sends one statistic to the central machine. The detailed procedures are as follows:
\begin{itemize}
	\item On each machine $j$, we calculate  $R_{j,k_n}^{(1)}$, $R_{j,k_n}^{(2)}$, $R_{j,k_{\rho}}^{(1)}$, $R_{j,k_{\rho}}^{(2)}$, $R_{j,k_{\rho}}^{(3)}$ and  estimate the second order parameter $\rho$ by \eqref{Eq: local rho estimation}. 
	\item On each machine $j$, we estimate the extreme value index by 
	$$
	\tilde{\gamma}_{j,k_n,k_{\rho},\tau}:=R_{j,k_n}^{(1)}-\frac{R_{j,k_n}^{(2)}-2\suit{R_{j,k_n}^{(1)}}^2}{2R_{j,k_n}^{(1)}\hat{\rho}_{j,k_{\rho},\tau} (1-\hat{\rho}_{j,k_{\rho},\tau})^{-1}},
	$$
	and   transmit   the estimates $\tilde{\gamma}_{j,k_n,k_{\rho},\tau}$ to the central machine.
	\item  On the central machine, we take the average of these estimates  by 
	$$
	\tilde{\gamma}_{k_n,k_{\rho},\tau}^{(3)}: =\frac{1}{m}\sum_{j=1}^m	\tilde{\gamma}_{j,k_n,k_{\rho},\tau}.
	$$		
\end{itemize}

The asymptotic theories of these two estimators $\tilde{\gamma}_{k_n,k_{\rho},\tau}^{(2)}$ and $\tilde{\gamma}_{k_n,k_{\rho},\tau}^{(3)}$ are left for further study. We only provide a finite sample comparison between the proposed estimator and these two estimators. 

In this comparison, we fix $\tau=0.5$. Figure \ref{figure frechet simu2} shows the RMSE for the  Fr\'echet distribution. The figures  for the Burr distribution and the absolute Cauchy distribution  have similar patterns and are therefore omitted. 
We observe that  all three bias corrected  estimators $\tilde{\gamma}_{k_n,k_{\rho},\tau}$, $\tilde{\gamma}_{k_n,k_{\rho},\tau}^{(2)}$ and $\tilde{\gamma}^{(3)}_{k_n,k_{\rho},\tau}$ generally perform better than  the original distributed Hill estimator. In addition, 
$\tilde{\gamma}_{k_n,k_{\rho},\tau}$ and $\tilde{\gamma}^{(2)}_{k_n,k_{\rho},\tau}$ have similar performance for all three values of $m$ with $\tilde{\gamma}_{k_n,k_{\rho},\tau}$ performing slightly better for the  Fr\'echet distribution and $\tilde{\gamma}^{(2)}_{k_n,k_{\rho},\tau}$ performing slightly better for the absolute Cauchy distribution. 
 
The performance of $\tilde{\gamma}^{(3)}_{k_n,k_{\rho},\tau}$ is unstable when $m$ is at a high level. In this case, $n$ is at a low level. Therefore, conducting bias correction  on each machine is suboptimal since the bias correction procedure requires a relatively large sample size.

\begin{figure}[htbp]
	\centering
	\subfigcapskip=-8pt
	\subfigure[m=1]{
	\includegraphics[width=0.33\textwidth]{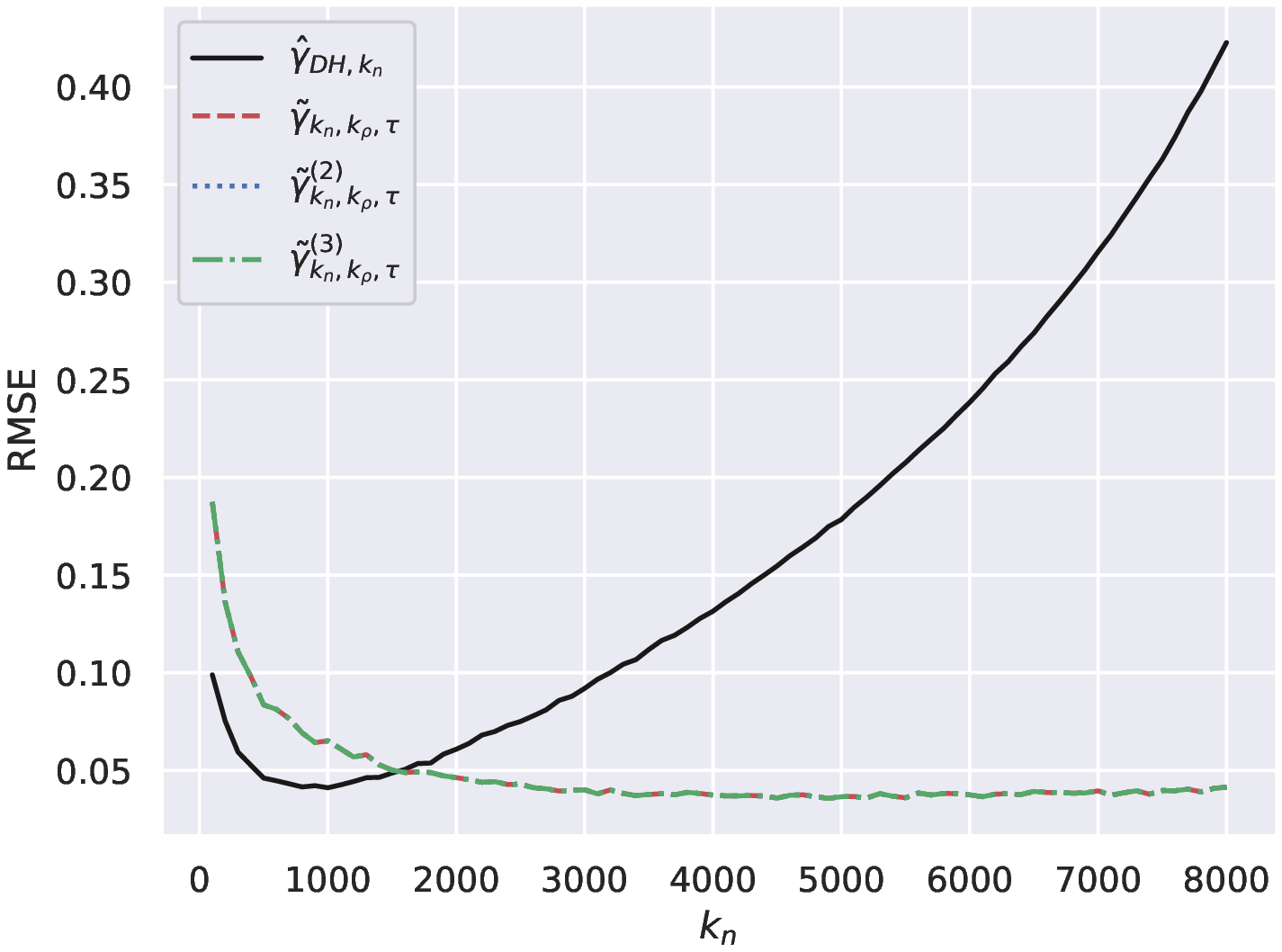}
	}
	% \vspace{-2ex}
	% \subfigure[m=10]{
	% \includegraphics[width=0.33\textwidth]{program/simu2/figs/Frechet_m10_0}
	% }
	\hspace{-7ex}
	\subfigure[m=20]{
	\includegraphics[width=0.33\textwidth]{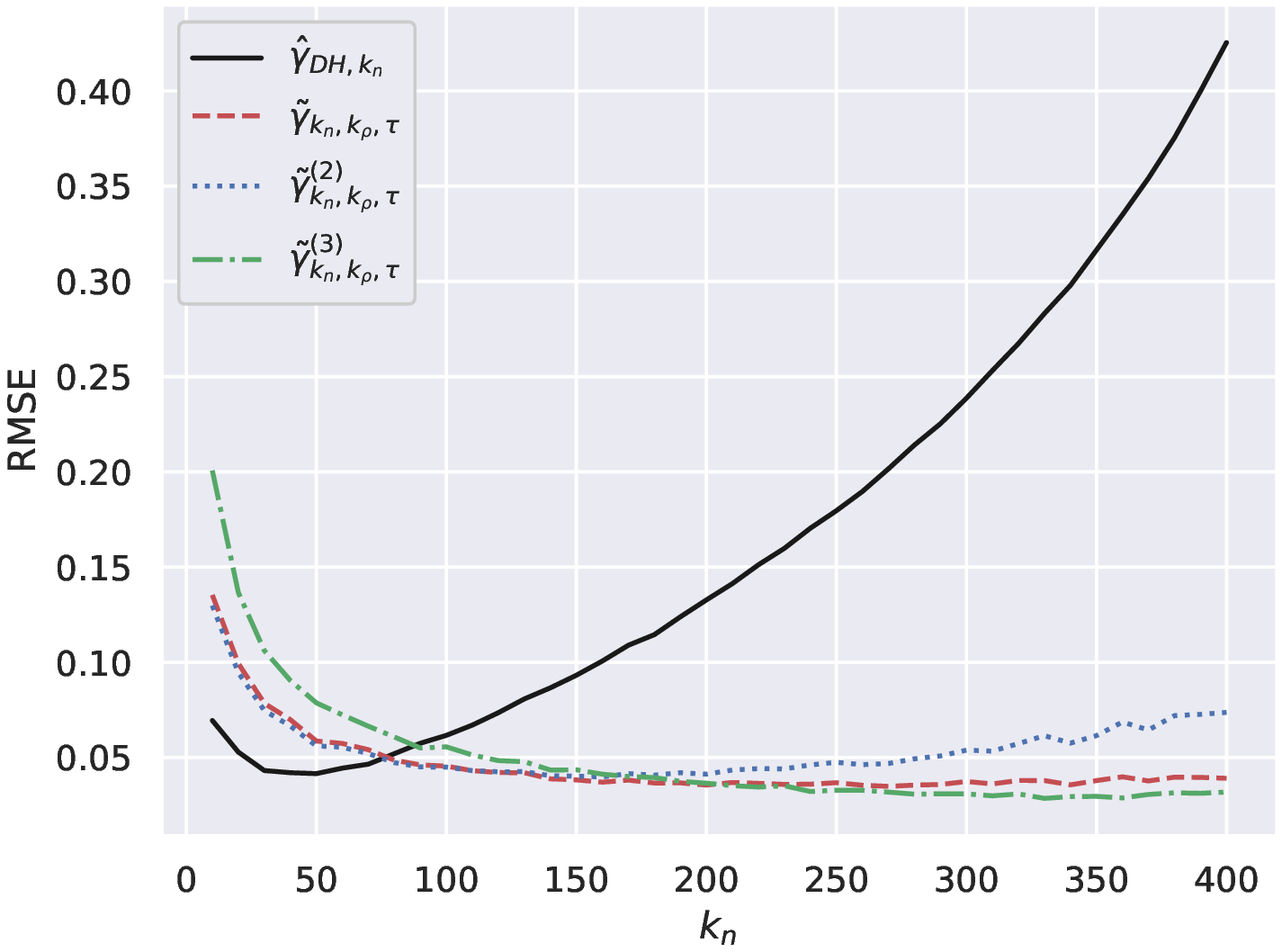}
	}
	\hspace{-7ex}
	\subfigure[m=100]{
		\includegraphics[width=0.33\textwidth]{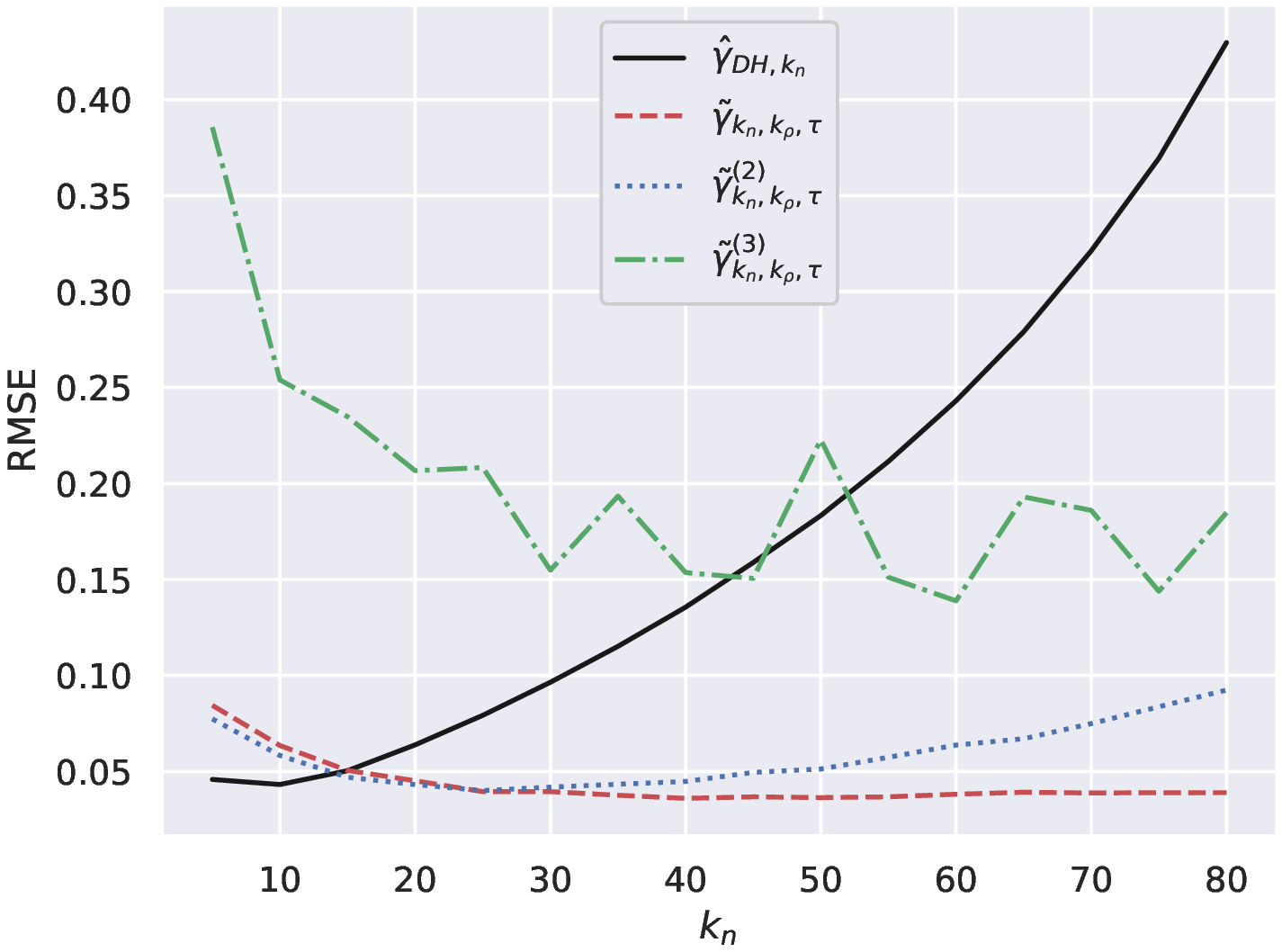}
		}
	\caption{RMSE for the Fr\'echet distribution.}
	\label{figure frechet simu2}
\end{figure}

\section{Discussion}
In this section, we discuss three extensions of our main results. The first two considers relaxing some technical assumptions in the current framework. The last one extends our result to estimating high quantiles.

First, we relax the assumption that the sample sizes on all machines are equal. Assume that $N$ observations are distributed stored in $m$ machines with $n_j = n_j(N)$ observations in  machine $j$, $j=1,2,\dots,m$, i.e. $N= \sum_{j=1}^m n_j$.  We assume that all $n_j, j=1,2,\dots,m$ diverge in the same order. Mathematically, there exist positive constants $c_1$ and $c_2$, such that for all $N\ge 1$,
$$
c_1\le \min_{1\le j \le m} n_j m/N \le \max_{1\le j \le m} n_j m/N \le c_2. 
$$

We choose $k_j, j=1,2,\dots,m$ such that the ratios $k_j/n_j$ are homogenous across all the $m$ machines, i.e.,
$$
k_1/n_1 = k_2/n_2=\cdots = k_m/n_m=:k/n,
$$
where $k = m^{-1}\sum_{j=1}^m k_j$ and $n =N/m$. 
Define 
$$
R_{k}^{(\alpha)}: =\sum_{j=1}^m \frac{n_j}{N} R_{j,k}^{(\alpha)}, \quad \alpha = 1,2,3.  
$$

  Under the same conditions as in Proposition \ref{theorem for expansion},  by following similar steps as in the proof of the proposition,    we can obtain that, as $N\to\infty$,
$$
\begin{aligned}
	&\sqrt{km}\suit{R_{k}^{(1)} -\gamma} \\
	 &=\gamma P_N^{(1)} +\sqrt{km}A_0(n/k) \frac{1}{m}\sum_{j=1}^m \frac{g(k_j,n_j,\rho)}{1-\rho} +\sqrt{km}A_0(n/k)B_0(n/k) \frac{1}{m}\sum_{j=1}^m \frac{g(k_j,n_j,\rho+\tilde{\rho})}{1-\rho-\tilde{\rho}}+o_P(1).
\end{aligned}
$$
Similar results hold for $R_k^{(2)}$ and $R_k^{(3)}$.

Then, with defining  the asymptotically unbiased
distributed estimator for the extreme value index as
$$
	\tilde{\gamma}_{k_n,k_{\rho},\tau}:=R_{k_{n}}^{(1)}-\frac{R_{k_{n}}^{(2)}-2\suit{R_{k_{n}}^{(1)}}^2}{2R_{k_{n}}^{(1)}\hat{\rho}_{k_{\rho},\tau} (1-\hat{\rho}_{k_{\rho},\tau})^{-1}},
$$
Theorem \ref{Theorem : gamma} still holds.

Second, we relax the assumption that all the data are drawn from the same distribution. We maintain the assumption that observations on the same machine follow the same distribution, but assume that observations across machines are not identically distributed. More specifically, denote the common distribution function of the observations in machine $j$ as $F_{m,j}, j=1,2,\dots,m$. We assume the heteroscedastic extreme model in \cite{einmahl2016statistics} holds for $F_{m,j}, j=1,2,\dots,m$: there exists a continous distribution function $F$ such that 
\begin{equation}\label{heter}
	\lim_{x\to\infty} \frac{1-F_{m,j}(x)}{1-F(x)} = c_{m,j},
\end{equation}
uniformly for all $1\le j\le m$ and all $m\in \mathbb{N}$ with $c_{m,j}$ uniformly bounded away from $0$ and $\infty$. 

Under this heteroscedastic extremes setup, the first order parameters $\gamma$ for all $F_{m,j}$, $ j=1,2,\dots,m$ are the same.  This heteroscedastic extreme setup is similar to the setup in Section 3 in \cite{chen2021distributed}. Its practical relevance can be again illustrated by the example of estimating tail risks in insurance claims. For a given type of insurance, claims in different insurance companies may not follow the same distribution due to the fact that different companies may be specialized in different segments of the market. Nevertheless, they may share the same shape parameter of the tail due to the underlying nature of the insured risk.

\cite{chen2021distributed} introduces additional assumptions to ensure that the heteroscedastic extremes assumption does not introduce an additional bias; see assumptions in Theorem 4 therein, particularly Condition D. Under the same assumption, by following similar techniques in the proof, we can show that the heteroscedastic extremes setup does not affect the statement in Theorem 2.

Third, we discuss how to obtain the \df for the high quantile $x(p_N) := U(1/p_N)$, where $p_N=O(1/N)$ as $N\to\infty$.
Motivated by \cite{de2016adapting},  we define the asymptotically unbiased distributed estimator for high quantile as 
$$
\hat{x}_{k_n, k_{\rho}, \tau}(p_N):=\frac{1}{m}\sum_{j=1}^m M_j^{(k_n+1)} \suit{\frac{k}{np_N}}^{\hat{\gamma}_{k_n,k_{\rho},\tau}} \suit{1-\frac{\suit{R_{k_n}^{(2)} - \suit{R_{k_n}^{(1)}}^2}\suit{1-\hat{\rho}_{k_{\rho},\tau}}^2}{2 R_{k,n}^{(1)} \suit{\hat{\rho}_{k_{\rho},\tau}}^2}}.
$$
Note that, the estimator $\hat{x}_{k_n, k_{\rho}, \tau}$ also adheres to a DC algorithm since each machine only sends \textit{six} values $\set{R_{j,k_n}^{(1)},R_{j,k_n}^{(2)},R_{j,k_{\rho}}^{(1)},R_{j,k_{\rho}}^{(2)},R_{j,k_{\rho}}^{(3)},M_j^{(k_n+1)}}$ to the central machine. Since $\hat{x}_{k_n, k_{\rho}, \tau}(p_N)$  are constructed  by $R_{k}^{(\alpha)}$ ($k=k_n$ and $ k_{\rho}$, $\alpha = 1,2,3$) and $m^{-1}\sum_{j=1}^m M_j^{(k_n+1)}$, the asymptotic theory of $\hat{x}_{k_n, k_{\rho}, \tau}(p_N)$  can be established using similar techniques as in the  proof of Theorem 4.2 in \cite{de2016adapting}. We leave the details to the readers.

\appendix
\section{Proofs}

\subsection{Preliminary}

\begin{lemma}\label{Lemma for Renyi}
	Let $Y, Y_1,\dots, Y_n$ be i.i.d.  Pareto (1)  random variables with distribution function $1-1/y, \ y\ge 1.$ Let $Y^{(1)} \ge \cdots \ge Y^{(n)}$ be the order statistics of $\set{Y_1,\dots,Y_n}$. Let $f$ be a function such that $\text{Var} \set{f(Y)}<\infty$. Then for any $k\ge 1$,
	$$
	\frac{1}{k}\sum_{i=1}^k f\suit{\frac{Y^{(i)}} {Y^{(k+1)}} }\stackrel{d}{=}\frac{1}{k}\sum_{i=1}^k f(Y_i^*),
	$$
	where $Y_1^*, Y_2^*,\ldots, Y_k^*$ are i.i.d. Pareto (1) random variables.
	Moreover, 
	$$
\sqrt{k}\set{\frac{1}{k}\sum_{i=1}^k f\suit{\frac{Y^{(i)}} {Y^{(k+1)}} }-\mathbb{E} f(Y) }
	$$
	is independent of $Y^{(k+1)}$ and asymptotically normally distributed with mean zero and variance $\text{Var} \set{f(Y)}$ as $n \to\infty$, provided that $k=k(n)\to \infty$ and $k/n \to 0$.
\end{lemma}
\begin{proof}[Proof of Lemma \ref{Lemma for Renyi}]
	This Lemma follows directly from Lemma 3.2.3 in \citet{de2006extreme} with the fact that $\log Y$ follows a standard exponential distribution.
\end{proof}

\begin{lemma}\label{Lemma for Expectation}
	Let $Y_1,\dots,Y_n$ be i.i.d.  Pareto (1) random variables and $Y^{(1)}\ge \cdots \ge Y^{(n)}$ be the  order statistics of $\set{Y_1,\dots,Y_n}$. Then for any $\rho<0$,
	$$
\mathbb{E}\set{\suit{\frac{k}{n}Y^{(k+1)}}^{\rho}}=g(k,n,\rho),
	$$
	where $g(k,n,\rho)$ is defined in \eqref{Def of g(k,n,rho)}.
	Moreover, if $k$ is a fixed integer, then $g(k,n,\rho) \to k^{\rho}\Gamma(k-\rho+1)/\Gamma(k+1)$ as $n \to \infty$. If $k$ is an intermediate sequence, i.e. $k\to \infty, k/n \to 0$ as $n \to \infty$, then, 
	$$
g(k,n,\rho)=1+\frac{1}{2}(\rho^2-\rho)k^{-1}-\frac{1}{2}(\rho^2-\rho)(n-\rho)^{-1}+O(k^{-2}).
	$$

\end{lemma}

\begin{proof}[Proof of Lemma \ref{Lemma for Expectation}]
	
$$
\begin{aligned}
	\mathbb{E}\set{\suit{\frac{k}{n}Y^{(k+1)}}^{\rho}} &= \frac{n!}{(n-k-1)!k!} \int_{1}^{\infty} \suit{1-\frac{1}{y}}^{n-k-1}\suit{\frac{1}{y}}^{k+2} \suit{\frac{k}{n}y}^{\rho}dy \\
	&=	\suit{\frac{k}{n}}^{\rho} \frac{n!}{(n-k-1)!k!}\int_{1}^{\infty} \suit{1-\frac{1}{y}}^{n-k-1}\suit{\frac{1}{y}}^{k+2-\rho} dy \\
	&= \suit{\frac{k}{n}}^{\rho} \frac{\Gamma(n+1)\Gamma(k-\rho+1)}{\Gamma(n-\rho+1)\Gamma(k+1)} \\
	&=g(k,n,\rho).
\end{aligned}
$$
We first handle the case when $k$ is a fixed integer. By the Stirling's formula,
$$
\Gamma(x)=\sqrt{2\pi(x-1)}\set{e^{-1}(x-1)}^{x-1}\set{1+(x-1)^{-1}/12+O(1/x^2)}
$$ 
as $x\to \infty$, we have that, as $n \to \infty$,
$$
\begin{aligned}
\Gamma(n+1)\sim (2\pi n)^{1/2}\suit{\frac{n}{e}}^{n}, \quad \Gamma(n-\rho+1)\sim \set{2\pi(n-\rho)}^{1/2}\suit{\frac{n-\rho}{\rho}}^{n-\rho},
\end{aligned}
$$
which leads to  
$$
g(k,n,\rho)\to k^{\rho}\frac{\Gamma(k-\rho+1)}{\Gamma(k+1)}.
$$

Next, we handle the case when $k$ is an intermediate sequence. 	By the Stirling's formula,
	we have that, as $n \to \infty$,
$$
\begin{aligned}
g(k,n,\rho)&=\suit{1-\frac{\rho}{k}}^{k-\rho+1/2}\suit{1+\frac{\rho}{n-\rho}}^{n-\rho+1/2} \frac{1+n^{-1}/12+O(n^{-2})}{1+(n-\rho)^{-1}/12+O(n^{-2})} \frac{1+(k-\rho)^{-1}/12+O(k^{-2})}{1+k^{-1}/12+O(k^{-2})}\\
&=\suit{1-\frac{\rho}{k}}^{k-\rho+1/2}\suit{1+\frac{\rho}{n-\rho}}^{n-\rho+1/2}\set{1+O(n^{-2})}\set{1+O(k^{-2})}.\\
\end{aligned}
$$ 
By the Taylor's formula and some direct calculation, we obtain that, as $n\to \infty$,
$$
\suit{1-\frac{\rho}{k}}^{k-\rho+1/2}  = e^{-\rho}\set{1+\frac{1}{2}(\rho^2-\rho)k^{-1}+O(k^{-2})},  
$$
and 
$$
\suit{1+\frac{\rho}{n-\rho}}^{n-\rho+1/2} =e^{\rho}\set{1-\frac{1}{2}(\rho^2-\rho)(n-\rho)^{-1}+O(n^{-2})}.
$$
It follows that, as $n \to \infty$,
$$
g(k,n,\rho)=1+\frac{1}{2}(\rho^2-\rho)k^{-1}-\frac{1}{2}(\rho^2-\rho)(n-\rho)^{-1}+O(k^{-2}).
$$

\end{proof}

\begin{lemma}\label{Lemma for Lyapunov}
	Let $Y_1,\dots,Y_n$ be i.i.d.  Pareto (1) random variables and $Y^{(1)}\ge \cdots \ge Y^{(n)}$ be the  order statistics of $\set{Y_1,\dots,Y_n}$.  Define for $\rho<0$, 
	$$
	Z_k=\frac{1}{k}\sum_{i=1}^k \frac{\suit{Y^{(i)}/Y^{(k+1)}}^{\rho}-1}{\rho}.
	$$
	Then, the following results hold. 
	\begin{itemize}
		\item[(i)] 	For fixed $k$, $\mathbb{E}(Z_k^a)<\infty$, for $a=1,2,3,4$. Moreover, $\mathbb{E}\suit{Z_k^2}-\set{\mathbb{E}\suit{Z_k}}^2>0$.
		\item[(ii)] For  intermediate $k$, i.e., $k = k(n)\to \infty, k/n \to 0$ as $n\to \infty$, and  $a=1,2,3,4$, 
		$$
		\mathbb{E}\suit{Z_k^a} = \frac{1}{(1-\rho)^a}\set{1+\frac{a(a-1)}{2(1-2\rho)}\frac{1}{k}+O(k^{-2})}.
		$$
	\end{itemize}
\end{lemma}
\begin{proof}[Proof of Lemma \ref{Lemma for Lyapunov}]
	By Lemma \ref{Lemma for Renyi}, we have that,
	$$
	Z_k \stackrel{d}{=} \frac{1}{k}\sum_{i=1}^k \frac{\suit{Y_i^*}^{\rho}-1}{\rho},
	$$
	where $Y_1^*,\dots,Y_k^*$ are i.i.d. Pareto (1) random variables. Denote $T_i=\set{(Y_i^*)^{\rho}-1}/\rho$, for $i=1,\dots,k$ and $Z_k= k^{-1}\sum_{i=1}^k T_i$. Then, $T_i, i=1,\dots,k$ follows the generalized Pareto distribution with the cumulative distribution function $F(t) = 1-(1+\rho t)^{-1/\rho}$. Thus, we have that for $a=1,2,3,4$,
	$$
\mathbb{E}(T_i^a)=\frac{a!}{(1-a\rho)\cdots(1-\rho)}.
	$$

First, we handle the case when $k$ is fixed. The result is obvious  since $kZ_k$ is a finite sum of i.i.d. generalized Pareto random variables with shape parameter $\rho<0$.

Next, we handle the case when $k$ is an intermediate sequence.  For $a=1$, we have that, $E(Z_k)=E(T_i)=(1-\rho)^{-1}$. 

For $a=2$, we have that,
$$
\begin{aligned}
	\mathbb{E}\suit{Z_k^2} &= \frac{1}{k^2}\set{\sum_{i=1}^k E\suit{T_i^2} +\sum_{i\ne j}\mathbb{E} \suit{T_i} \mathbb{E}\suit{T_j}} \\
	&=\frac{1}{k^2}\bra{k E\suit{T_i^2} +k(k-1) \set{\mathbb{E} \suit{T_i}}^2} \\
	& = \frac{1}{(1-\rho)^2}+\frac{1}{k}\frac{1}{(1-2\rho)(1-\rho)^2}.
\end{aligned}
$$
 
For $a=3$, we have that 
$$
\begin{aligned}
	\mathbb{E}\suit{Z_k^3} &= \frac{1}{k^2}\set{\sum_{i=1}^k \mathbb{E}\suit{T_i^3} +\sum_{i = j \ne l}\mathbb{E} \suit{T_i T_j} \mathbb{E}\suit{T_l}+ \sum_{i\ne j \ne l}\mathbb{E} \suit{T_i} \mathbb{E}\suit{T_j}\mathbb{E}\suit{T_l}} \\
	&=\frac{1}{k^3}\bra{k \mathbb{E}\suit{T_i^3} +3k(k-1) \mathbb{E} \suit{T_i^2} \mathbb{E}\suit{T_i}  +k(k-1)(k-2) \set{\mathbb{E}(T_i)}^3}\\
	& = \frac{1}{(1-\rho)^3}+\frac{1}{k}\frac{3}{(1-2\rho)(1-\rho)^3}+O(k^{-2}).
\end{aligned}
$$

The term  $\mathbb{E} \suit{Z_k^4}$ can be handled in a similar way as that for handling $\mathbb{E}\suit{Z_k^3}$.
\end{proof}

\begin{lemma}\label{lemma for third order equality}
	Assume that the distribution function $F$ satisfies the third order condition \eqref{third order condition}. Then there exist two functions $A_0(t)\sim A(t)$ and $B_0(t)=O\set{B(t)}$ as $t\to\infty$, such that for any $\delta>0$, there exists  a $t_0 = t_0(\delta)>0$, for all $t\ge t_0$ and $tx\ge t_0$,
	$$
\abs{\frac{\frac{\log U(tx)-\log U(t)-\gamma \log x}{A_0(t)}-\frac{x^{\rho}-1}{\rho}}{B_0(t)}-\frac{x^{\rho+\tilde{\rho}}-1}{\rho+\tilde{\rho}}}\le \delta x^{\rho+\tilde{\rho}}\max(x^{\delta},x^{-\delta}).
	$$
\end{lemma}

\begin{proof}[Proof of Lemma \ref{lemma for third order equality}]
	This lemma follows from applying Theorem B.3.10 in \citet{de2006extreme} to the function $f(t):=\log U(t)-\gamma \log t$.
\end{proof}

\subsection{Proofs for Section 3}
Recall that $U=\set{1/(1-F)}^{\leftarrow}$. Then $X\stackrel{d}{=}U(Y)$, where $Y$ follows the Pareto (1) distribution. Since we have i.i.d. observations $\set{X_1,\dots,X_N}$, we can write $X_i\stackrel{d}{=}U(Y_i)$, where $\set{Y_1,\dots,Y_N}$ is a random sample of $Y$. Recall that the $N$ observations are stored in $m$ machines with $n$ observations each. For machine $j$, let $Y_j^{(1)} \ge \cdots\ge Y_j^{(n)}$ denote the order statistics of the $n$ Pareto (1) distributed variables corresponding to the $n$ observations in this machine. Then $M_j^{(i)}\stackrel{d}{=}U(Y_j^{(i)}), i=1,\dots,n, j=1,\dots,m$.

\begin{proof}[Proof of Proposition \ref{theorem for expansion}]
	 We intend to replace $t$ and $tx$ in Lemma \ref{lemma for third order equality} by $n/k$ and $Y_j^{(i)}, i=1,\dots,k+1, j=1,\dots,m$, respectively. For this purpose, we introduce the set
	$$
	\mathcal{F}_{t_0}:=\set{Y_j^{(k+1)}\ge t_0, \ for \ all \ 1\le j\le m}.
	$$
By Lemma S.2 in the supplementary material of \citet{chen2021distributed}, we have that for any $t_0>1$, if condition \eqref{Condtion of m,n} holds, then
$
\lim_{N\to \infty}\mathbb{P}\suit{\mathcal{F}_{t_0}}=1.
$
Then, we can apply the intended replacement to get that, as $N\to \infty$,
\begin{equation}\label{Third order expansion}
\begin{aligned}
\log U(Y_j^{(i)}) -\log U(n/k)& =-\gamma \log \suit{kY_j^{(i)}/n} -A_0(n/k)\set{\suit{kY_j^{(i)}/n}^{\rho}-1}/\rho \\
&\quad + A_0(n/k)B_0(n/k)\set{\suit{kY_j^{(i)}/n}^{\rho+\tilde{\rho}}-1}/\suit{\rho+\tilde{\rho}}\\
& \quad +o_P(1)A_0(n/k)B_0(n/k)\suit{kY_j^{(i)}/n}^{\rho+\tilde{\rho} \pm \delta},
\end{aligned}
\end{equation}
where the $o_P(1)$ term is uniform for all $1\le i\le k+1$ and $1\le j\le m$.
By applying \eqref{Third order expansion} twice for a general $i$ and $i=k+1$ and the inequality $x^{\rho\pm \delta}/y^{\rho \pm \delta}\le (x/y)^{\rho \pm \delta}$ for any $x,y>0$, we get that as $N \to \infty$,
	\begin{equation}\label{2_Third order expansion}
	\begin{aligned}
	&\quad \log U\left(Y_{j}^{(i)}\right) -\log U\left(Y_{j}^{(k+1)}\right) \\
	 & \quad =\gamma\left(\log Y_{j}^{(i)}-\log Y_{j}^{(k+1)}\right) \\
	&\qquad + A_{0}(n / k) \left(k Y_{j}^{(k+1)} / n\right)^{\rho}\set{\left( Y_{j}^{(i)} / Y_{j}^{(k+1)} \right)^{\rho}-1}/{\rho}\\
	&\qquad  + A_{0}(n / k) B_{0}(n / k)\left(k Y_{j}^{(k+1)} / n\right)^{\rho+\tilde{\rho}}\set{\left( Y_{j}^{(i)} / Y_{j}^{(k+1)} \right)^{\rho+\tilde{\rho}}-1}/\suit{\rho+\tilde{\rho}} \\
	&\qquad + o_P(1)A_{0}(n / k)B_0(n/k) \left(k Y_{j}^{(k+1)}/n\right)^{\rho+\tilde{\rho} \pm \delta}\set{\left( Y_{j}^{(i)}/Y_j^{(k+1)}\right)^{\rho +\tilde{\rho} \pm \delta}+1}.
	\end{aligned}
\end{equation}
By taking the average across $i$ and $j$, we obtain that
$$
\begin{aligned}
&\sqrt{km}\suit{R_k^{(1)}-\gamma}	\\
&\quad=  \gamma \sqrt{km}\frac{1}{m}\frac{1}{k}\sum_{j=1}^m \sum_{i=1}^k \set{\log \suit{Y_j^{(i)}/Y_j^{(k+1)}}-\gamma}\\
&\qquad + \sqrt{km}A_0(n/k)\frac{1}{m}\sum_{j=1}^m \left(k Y_{j}^{(k+1)} / n\right)^{\rho} \rho^{-1} \frac{1}{k}\sum_{i=1}^k \set{\left( Y_{j}^{(i)} / Y_{j}^{(k+1)} \right)^{\rho}-1}\\
& \qquad + \sqrt{km}A_0(n/k)B_0(n/k)\frac{1}{m}\sum_{j=1}^m \left(k Y_{j}^{(k+1)} / n\right)^{\rho+\tilde{\rho}} (\rho+\tilde{\rho})^{-1}\frac{1}{k}\sum_{i=1}^k \set{\left( Y_{j}^{(i)} / Y_{j}^{(k+1)} \right)^{\rho+\tilde{\rho}}-1}\\
&\qquad +o_P(1)\sqrt{km}A_{0}(n / k)B_0(n/k) \frac{1}{m}\sum_{j=1}^m \left(k Y_{j}^{(k+1)} / n\right)^{\rho+\tilde{\rho}\pm \delta}\frac{1}{k}\sum_{i=1}^k \set{\left( Y_{j}^{(i)}/Y_j^{(k+1)}\right)^{\rho +\tilde{\rho} \pm \delta}+1}\\
&\quad=: I_1+I_2+I_3+I_4.
\end{aligned}
$$
Firstly, we handle $I_1$. By Lemma \ref{Lemma for Renyi}, we have that, 
$$
I_1 \stackrel{d}{=} \gamma\sqrt{km}\suit{ \frac{1}{km}\sum_{j=1}^m\sum_{i=1}^k \log Y_{i}^{j,*}-1},
$$
where $Y_{i}^{j,*}, i=1,\dots,k,j=1,\dots,m$ are independent and identically distributed  Pareto (1) random variables.
The central limit theorem yields that as $N \to \infty$,
$
	I_1 = \gamma P_N^{(1)} +o_P(1),
$
where $P_N^{(1)}\sim N(0,1)$.

For $I_2$,  write 
$\delta_{j,n}=\left(k Y_{j}^{(k+1)} / n\right)^{\rho}(k\rho)^{-1}\sum_{i=1}^k \set{\left( Y_{j}^{(i)} / Y_{j}^{(k+1)} \right)^{\rho}-1}$.  Then we have that $I_2=\sqrt{km}A_0(n/k)m^{-1}\sum_{j=1}^m \delta_{j,n}$, where $\delta_{j,n}, j=1,\dots,m$ are i.i.d. random variables.

We are going to show that, as $N \to \infty$,
\begin{equation}\label{CLT for Triangular}
	\sqrt{km} \set{\frac{1}{m}\sum_{j=1}^m \delta_{j,n} -\mathbb{E}\suit{\delta_{j,n}} }=O_P(1).
\end{equation}

If $k$ is fixed, \eqref{CLT for Triangular} follows directly from Lemma \ref{Lemma for Lyapunov} (i) and the  Lyapunov central limit theorem for triangular array.

Next, we handle the case when $k$ is an intermediate sequence. In this case, in order to apply the  Lyapunov central limit theorem with $4$-th moment, we need to calculate $\text{Var}\suit{\delta_{j,n}}$ and $\mathbb{E}\bra{\set{\delta_{j,n}-\mathbb{E}\suit{\delta_{j,n}}}^4}$.
Denote $m_{n}^{(a)}:=\mathbb{E}\set{\suit{\delta_{j,n}}^a},\ a=1,2,3,4$. 
By Lemma \ref{Lemma for Renyi}, we have that, 
$$
m_n^{(a)}= g(k,n,a\rho)\mathbb{E} \bra{\set{\frac{1}{k}\sum_{i=1}^k\frac{\suit{Y_{j}^{(i)} / Y_{j}^{(k+1)}}^{\rho}-1}{\rho}}^a }.
$$

First, we calculate $\text{Var}\suit{\delta_{j,n}}$. By Lemma \ref{Lemma for Lyapunov}, we have that,
$$
\begin{aligned}
\text{Var}(\delta_{j,n})&=m_{n}^{(2)}-\suit{m_{n}^{(1)}}^2  \\
&= g(k,n,2\rho)\set{\frac{1}{(1-\rho)^2}+\frac{1}{k}\frac{1}{(1-2\rho)(1-\rho)^2}+O(k^{-2})}-\set{g(k,n,\rho)}^2\set{\frac{1}{(1-\rho)^2}+O(k^{-2})}\\
&= \frac{1}{k} g(k,n,2\rho)\frac{1}{(1-2\rho)(1-\rho)^2}+\bra{g(k,n,2\rho)-\set{g(k,n,\rho)}^2}\frac{1}{(1-\rho)^2}+O(k^{-2}),
\end{aligned}
$$
here in the last step, we used the fact that as $n\to \infty$, $g(k,n,\rho)\to 1$ and $g(k,n,2\rho)\to 1$.
By Lemma \ref{Lemma for Expectation}, we have that, as $n \to \infty$, 
$$
\begin{aligned}
g(k,n,2\rho)-\set{g(k,n,\rho)}^2&=1+\frac{1}{2}\suit{4\rho^2-2\rho}\frac{1}{k} +o(k^{-1})-\set{1+\frac{1}{2}\suit{\rho^2-\rho}\frac{1}{k}+o(k^{-1})}^2=\frac{1}{k}\rho^2+o(k^{-1}).
\end{aligned}
$$ 
Hence, as $n\to \infty$, $\text{Var}\suit{\delta_{j,n}}=k^{-1}(1-\rho)^{-2}\suit{\suit{1-2\rho}^{-1}+\rho^2}+o(k^{-1})$.
 
 Next, we calculate $\mathbb{E}\bra{\set{\delta_{j,n}-\mathbb{E}\suit{\delta_{j,n}}}^4}$.
 By Lemma \ref{Lemma for Expectation} and Lemma \ref{Lemma for Lyapunov}, we have that, for $a=3,4$, as $N\to \infty$, 
$$
\begin{aligned}
m_{n}^{(a)}&=(1-\rho)^{-a}\set{1+\frac{1}{2}\frac{1}{k}\frac{a(a-1)}{1-2\rho}+O(k^{-2})}\set{1+\frac{1}{2}(a^2\rho^2-a\rho)k^{-1}-\frac{1}{2}(a^2\rho^2-a\rho)(n-a\rho)^{-1}+O(k^{-2})}\\
&=(1-\rho)^{-a}\set{1+k^{-1}\frac{1}{2}\frac{a(a-1)}{1-2\rho}+\frac{1}{2}(a^2\rho^2-a\rho)k^{-1}-\frac{1}{2}(a^2\rho^2-a\rho)(n-a\rho)^{-1}+O(k^{-2})}.
\end{aligned}
$$
Note that, 
$$
\mathbb{E}\bra{\set{(\delta_{j,n}-\mathbb{E}\suit{\delta_{j,n}}}^4}=m_n^{(4)}-4m_n^{(3)}m_n^{(1)}+6m_n^{(2)}\suit{m_n^{(1)}}^2-3\suit{m_n^{(1)}}^4.
$$
By some direct calculation, all terms of order $k^{-1}$ and $n^{-1}$ are cancelled  out. 
Thus, as $N \to \infty$, $\mathbb{E}\bra{\set{(\delta_{j,n}-\mathbb{E}\suit{\delta_{j,n}}}^4} =O(k^{-2}). $
Combining $\text{Var}(\delta_{j,n})$ and $\mathbb{E}\bra{\set{\delta_{j,n}-\mathbb{E}\suit{\delta_{j,n}}}^4}$,  we conclude that the sequences $\set{\delta_{j,n}}_{j=1}^m$  satisfy the Lyapunov's condition.  Then, \eqref{CLT for Triangular} follows by the central limit theorem.
Applying \eqref{CLT for Triangular}, we obtain that, as $N\to \infty$,
$$
I_2 = \sqrt{km}A_0(n/k) \set{\mathbb{E}\suit{\delta_{j,n}}+O_P(1/\sqrt{km})}=\frac{g(k,n,\rho)}{1-\rho}\sqrt{km}A_0(n/k)+o_P(1).
$$

For $I_3$, by using the weak law of large numbers for triangular array, we have that, as $N \to \infty$,								
$$
\begin{aligned}
I_3 &= \frac{\sqrt{km}A_0(n/k)B_0(n/k)}{1-\rho-\tilde{\rho}}\mathbb{E}\set{\suit{kY_1^{(k+1)}/n}^{\rho+\tilde{\rho}}}\set{1+o_P(1)}\\
&=\sqrt{km}A_0(n/k)B_0(n/k)\frac{g(k,n,\rho+\tilde{\rho})}{1-\rho-\tilde{\rho} }+o_P(1),
\end{aligned}
$$
where the last equality follows by the condition  $\sqrt{km}A(n/k)B(n/k)=O(1)$.

For $I_4$, by similar arguments as for $I_3$, we obtain that, as $N\to \infty$, $I_4\stackrel{P}{\to} 0$.
Combining $I_1,I_2,I_3$ and $I_4$, we have proved (i).

Next, we handle $R_k^{(2)}$.
By \eqref{2_Third order expansion}, we obtain that, as $N\to\infty,$ 
$$
\begin{small}
\begin{aligned}
&\sqrt{km}\suit{R_k^{(2)}-2\gamma^2} \\
&= \gamma^2 \frac{1}{mk}\sum_{j=1}^m \sum_{i=1}^k \set{\log^2 \suit{Y_j^{(i)}/Y_j^{(k+1)}}-2}\\
& \quad + 2\gamma\sqrt{km}A_0(n/k)\frac{1}{km}\sum_{j=1}^m \suit{kY_j^{(k+1)}/n}^{\rho}\sum_{i=1}^k \log \suit{Y_j^{(i)}/Y_j^{(k+1)}}\set{\suit{Y_j^{(i)}/Y_j^{(k+1)}}^{\rho}-1}/\rho \\ 
& \quad +\sqrt{km} A_0^2(n/k) \frac{1}{km}\sum_{j=1}^m \suit{kY_j^{(k+1)}/n}^{2\rho}\sum_{i=1}^k \set{\suit{Y_j^{(i)}/Y_j^{(k+1)}}^{\rho}-1}^2/\rho^2 \\
& \quad +2\gamma\sqrt{km}A_0(n/k) B_0(n/k) \frac{1}{km}  \sum_{j=1}^m \suit{kY_j^{(k+1)}/n}^{\rho+\tilde{\rho}}\sum_{i=1}^k \log \suit{Y_j^{(i)}/Y_j^{(k+1)}}\frac{\suit{Y_j^{(i)}/Y_j^{(k+1)}}^{\rho+\tilde{\rho}}-1}{\rho+\tilde{\rho}}\\
&\quad + o_P(1)\\
&=: I_5+I_6+I_7+I_8+o_P(1).
\end{aligned}
\end{small}
$$

For $I_5$, by Lemma \ref{Lemma for Renyi}, we have that 
$$
I_5\stackrel{d}{=}\gamma^2\sqrt{km}\set{ \frac{1}{km}\sum_{j=1}^m\sum_{i=1}^k\suit{\log Y_{i}^{j,*}}^2-2}.
$$
The central limit theorem yields that as $N \to \infty$, 
$I_5 =  \gamma^2 P_N^{(2)}+o_P(1)$, where $P_N^{(2)}\sim N(0,20)$. 
 In addition, the covariance of $P_N^{(1)}$ and $P_N^{(2)}$ is equal to the covariance of $\log Y_i^{j,*}$ and $\suit{\log Y_i^{j,*}}^2$, where $Y_{i}^{j,*}$ follows the  Pareto (1) distribution. Hence,
$\text{Cov}(P_N^{(1)},P_N^{(2)})=4.$

For $I_6$, we write $I_6 = 2\sqrt{km}A_0(n/k)m^{-1}\sum_{j=1}^m \eta_{j,n}$, where
$$
\eta_{j,n}=\left(k Y_{j}^{(k+1)} / n\right)^{\rho}(k\rho)^{-1}\sum_{i=1}^k \log \left( Y_{j}^{(i)} / Y_{j}^{(k+1)} \right)\set{\left( Y_{j}^{(i)} / Y_{j}^{(k+1)} \right)^{\rho}-1}
$$
are i.i.d. random variables for $j=1,2,\dots,m$.	
We can verify the Lyapunov's condition for the series $\set{\eta_{j,n}}_{j=1}^m$  following similar steps as those for  $\set{\delta_{j,n}}_{j=1}^m$.  Then by applying the  central limit theorem and Lemma \ref{Lemma for Expectation}, we obtain that 
$$
I_6=2\gamma\sqrt{km}A_0(n/k) g(k,n,\rho)\frac{1}{\rho}\set{\frac{1}{(1-\rho)^2}-1}+o_P(1).
$$
By the weak law of large numbers for triangular array, we have that 
$$
I_7=\sqrt{km}A_0^2(n/k)\frac{g(k,n,2\rho)}{\rho^2}\set{\frac{1}{1-2\rho}-\frac{2}{1-\rho}+1}+o_P(1),
$$
and 
$$
I_8=2\gamma \sqrt{km}A_0(n/k)B_0(n/k)\frac{g(k,n,\rho+\tilde{\rho})}{\rho+\tilde{\rho}}\set{\frac{1}{(1-\rho-\tilde{\rho})^2}-1}+o_P(1).
$$
Combining the results for $I_5,I_6,I_7$ and $I_8$, we have proved (ii).

Finally, we handle $R_k^{(3)}$. Also, by \eqref{2_Third order expansion}, we have that
$$
\begin{small}
\begin{aligned}
	&\sqrt{km}\suit{R_k^{(3)}-6\gamma^3}\\
	 &= \gamma^3 \frac{1}{mk}\sum_{j=1}^m \sum_{i=1}^k \set{\log^3 \suit{Y_j^{(i)}/Y_j^{(k+1)}}-6}\\
	& \quad + 3\gamma^2\sqrt{km}A_0(n/k)\frac{1}{km}\sum_{j=1}^m \suit{kY_j^{(k+1)}/n}^{\rho}\sum_{i=1}^k \set{\log \suit{Y_j^{(i)}/Y_j^{(k+1)}}}^2\frac{\suit{Y_j^{(i)}/Y_j^{(k+1)}}^{\rho}-1}{\rho} \\ 
	& \quad + 3\gamma\sqrt{km}A_0^2(n/k)\frac{1}{km}\sum_{j=1}^m \suit{kY_j^{(k+1)}/n}^{2\rho}\sum_{i=1}^k \log \suit{Y_j^{(i)}/Y_j^{(k+1)}}\set{\frac{\suit{Y_j^{(i)}/Y_j^{(k+1)}}^{\rho}-1}{\rho}}^2\\
	& \quad +3\gamma^2\sqrt{km}A_0(n/k) B_0(n/k) \frac{1}{km}  \sum_{j=1}^m \suit{kY_j^{(k+1)}/n}^{\rho+\tilde{\rho}}\sum_{i=1}^k \set{\log \suit{Y_j^{(i)}/Y_j^{(k+1)}}}^2\frac{\suit{Y_j^{(i)}/Y_j^{(k+1)}}^{\rho+\tilde{\rho}}-1}{\rho+\tilde{\rho}}\\
	&\quad +o_P(1)\\
	&=: I_9+I_{10}+I_{11}+I_{12}+o_P(1).
	\end{aligned}
\end{small}
	$$
By similar steps as for handling the four items $I_5, I_6,I_7$ and $I_8$, we can show that 
 $I_9 =  \gamma^3 P_N^{(3)}+o_P(1) $,
 where $P_N^{(3)}\sim N(0,684)$ and
 $\text{Cov}(P_N^{(1)},P_N^{(3)})=18, \text{Cov}(P_N^{(2)},P_N^{(3)})=98. $ 
And 
$$
\begin{aligned}
I_{10}&=6\gamma^2 \sqrt{km}A_0(n/k)\frac{g(k,n,\rho)}{\rho}\set{\frac{1}{(1-\rho)^3}-1}+o_P(1),\\
I_{11}&=3\gamma\sqrt{km} A^2_0(n/k) \frac{g(k,n,2\rho)}{\rho^2}\set{\frac{1}{(1-2\rho)^2}-\frac{2}{(1-\rho)^2}+1}+o_P(1),\\
I_{12}&=6\gamma^2\sqrt{km}A_0(n/k)B_0(n/k)\frac{g(k,n,\rho+\tilde{\rho})}{\rho+\tilde{\rho}}\set{\frac{1}{(1-\rho-\tilde{\rho})^3}-1}+o_P(1),
\end{aligned}
$$
which yields (iii).
\end{proof}

\begin{proof}[Proof of Theorem \ref{Theorem of rho}]
	Applying Proposition \ref{theorem for expansion} with $k=k_{\rho}$, we have that, as $N \to \infty$,
	$$
\begin{aligned}
R_{k_{\rho}}^{(1)}&=\gamma+\frac{\gamma}{\sqrt{k_{\rho}m}} P_N^{(1)} + \frac{g(k_{\rho},n,\rho)}{1-\rho}A_0(n/k_{\rho})+\frac{g(k_{\rho},n,\rho+\tilde{\rho})}{1-\rho-\tilde{\rho}}A_0(n/k_{\rho})B_0(n/k_{\rho})+\frac{1}{\sqrt{k_{\rho}m}}o_P(1),\\
	R_{k_{\rho}}^{(2)}& =2\gamma^2+\frac{\gamma^2}{\sqrt{k_{\rho}m}} P_N^{(2)} +2 \gamma A_0(n/k_{\rho})\frac{g(k_{\rho},n,\rho)}{\rho}\set{\frac{1}{(1-\rho)^2}-1}\\
	&+A_0^2(n/k_{\rho})\dfrac{g(k_{\rho},n,2\rho)}{\rho^2}\suit{\frac{1}{1-2\rho}-\frac{2}{1-\rho}+1}   \\
		&+2\gamma A_0(n/k_{\rho})B_0(n/k_{\rho})\frac{g(k_{\rho},n,\rho+\tilde{\rho})}{\rho+\tilde{\rho}}\set{\frac{1}{(1-\rho-\tilde{\rho})^2}-1}+\frac{1}{\sqrt{k_{\rho}m}}o_P(1),\\
	R_{k_{\rho}}^{(3)}&= 6\gamma^3+\frac{\gamma^3}{\sqrt{k_{\rho}m}}P_N^{(3)}+6\gamma A_0(n/k_{\rho})\frac{g(k_{\rho},n,\rho)}{\rho}\set{\frac{1}{(1-\rho)^3}-1} \\
	&\quad  +3A_0^2(n/k_{\rho})\frac{g(k_{\rho},n,2\rho)}{\rho^2}\set{\frac{1}{(1-2\rho)^2}-\frac{2}{(1-\rho)^2}+1}\\
	&\quad +6\gamma A_0(n/k_{\rho})B_0(n/k_{\rho})\frac{g(k_{\rho},n,\rho+\tilde{\rho})}{\rho+\tilde{\rho}}\set{\frac{1}{(1-\rho-\tilde{\rho})^3}-1}+\frac{1}{\sqrt{k_{\rho}m}}o_P(1).
	\end{aligned}
	$$
As a consequence, we have that, as $N \to\infty$,  
$$
\begin{aligned}
	\suit{R_{k_{\rho}}^{(1)}}^{\tau}&=\gamma^{\tau}\set{1+ \frac{\tau}{\sqrt{k_{\rho}m}} P_N^{(1)} +\frac{\tau}{\gamma}\frac{g(k_{\rho},n,\rho)}{1-\rho} A_0(n/k_{\rho})+\frac{\tau}{\gamma}\frac{g(k_{\rho},n,\rho+\tilde{\rho})}{1-\rho-\tilde{\rho}}A_0(n/k_{\rho})B_0(n/k_{\rho})} \\
	&\quad +\frac{1}{\sqrt{k_{\rho}m}}o_P(1),\\
	\suit{R_{k_{\rho}}^{(2)}/2}^{\tau/2}& =\gamma^{\tau}\Bigg[1+\frac{\tau}{\sqrt{k_{\rho}m}} P_N^{(2)} +\frac{\tau}{2\gamma}A_0(n/k_{\rho})\frac{g(k_{\rho},n,\rho)}{\rho}\set{\frac{1}{(1-\rho)^2}-1}\\
	&\quad + \frac{\tau}{4\gamma} A_0^2(n/k_{\rho})\frac{g(k_{\rho},n,2\rho)}{\rho^2}\suit{\frac{1}{1-2\rho}-\frac{2}{1-\rho}+1}\\
	&\quad  +\frac{\tau}{2\gamma}A_0(n/k_{\rho})B_0(n/k_{\rho})\frac{g(k_{\rho},n,\rho+\tilde{\rho})}{\rho+\tilde{\rho}}\set{\frac{1}{(1-\rho-\tilde{\rho})^2}-1}\Bigg]+\frac{1}{\sqrt{k_{\rho}m}}o_P(1),\\
	\suit{R_{k_{\rho}}^{(3)}/6}^{\tau/3}&=\gamma^{\tau}\Bigg[1+\frac{\tau}{\sqrt{k_{\rho}m}}P_N^{(3)}+\frac{\tau}{3\gamma} A_0(n/k_{\rho})\frac{g(k_{\rho},n,\rho)}{\rho}\set{\frac{1}{(1-\rho)^3}-1} \\
	&\quad  +\frac{\tau}{6\gamma} A_0^2(n/k_{\rho})\frac{g(k_{\rho},n,2\rho)}{\rho^2}\set{\frac{1}{(1-2\rho)^2}-\frac{2}{(1-\rho)^2}+1}\\
	&\quad +\frac{\tau}{3\gamma} A_0(n/k_{\rho})B_0(n/k_{\rho})\frac{g(k_{\rho},n,\rho+\tilde{\rho})}{\rho+\tilde{\rho}}\set{\frac{1}{(1-\rho-\tilde{\rho})^3}-1}\Bigg ]+\frac{1}{\sqrt{k_{\rho}m}}o_P(1).
\end{aligned}
$$
It follows that, as $N\to\infty$,
$$
\begin{aligned}
	\gamma^{-\tau}\set{\suit{R_{k_{\rho}}^{(1)}}^{\tau} -\suit{R_{k_{\rho}}^{(2)}/2}^{\tau/2}} &= \frac{\tau}{\sqrt{k_{\rho}m}}\suit{P_N^{(1)}-P_N^{(2)}}+\frac{\tau}{\gamma}g(k_{\rho},n,\rho)A_0(n/k_{\rho})\frac{-\rho}{2(1-\rho)^2} \\
&+A_0^2\suit{n/k_{\rho}}O(1)+A_0\suit{n/k_{\rho}}B_0\suit{n/k_{\rho}}O(1)+\frac{1}{\sqrt{k_{\rho}m}}o_P(1),
\end{aligned}
$$
and
$$
\begin{aligned}
	\gamma^{-\tau}\set{\suit{R_{k_{\rho}}^{(2)}/2}^{\tau/2} -\suit{R_{k_{\rho}}^{(2)}/6}^{\tau/3}} &= \frac{\tau}{\sqrt{k_{\rho}m}}\suit{P_N^{(2)}-P_N^{(3)}}+\frac{\tau}{\gamma}g(k_{\rho},n,\rho)A_0(n/k_{\rho})\frac{\rho(\rho-3)}{6(1-\rho)^3} \\
	&+A_0^2\suit{n/k_{\rho}}O(1)+A_0\suit{n/k_{\rho}}B_0\suit{n/k_{\rho}}O(1)+\frac{1}{\sqrt{k_{\rho}m}}o_P(1).
\end{aligned}
$$
By the condition \eqref{k_rho_condition}, the dominating terms in the two expressions above are 
$$\frac{\tau}{\gamma}g(k_{\rho},n,\rho)A_0(n/k_{\rho})\frac{-\rho}{2(1-\rho)^2} 
\quad \text{and} 
\quad  \frac{\tau}{\gamma}g(k_{\rho},n,\rho)A_0(n/k_{\rho})\frac{\rho(\rho-3)}{6(1-\rho)^3},
$$
respectively. Therefore, as $N \to \infty$,
$$
\begin{aligned}
	T_{k_{\rho},\tau}&=3\frac{\rho-1}{\rho-3} \set{1+\frac{\gamma}{\sqrt{k_{\rho}m}}\frac{2(1-\rho)^2}{-\rho A_0(n/k_{\rho})}\suit{P_N^{(1)}-P_N^{(2)}}-\frac{\gamma}{\sqrt{k_{\rho}m}A_0(n/k_{\rho})} \frac{6(1-\rho)^3}{\rho^2-3\rho}\suit{P_N^{(2)}-P_N^{(3)}}}\\
	&\quad +O_P\set{A_0\suit{n/k_{\rho}}}+O_P\set{B_0\suit{n/k_{\rho}}}+\frac{1}{\sqrt{k_{\rho}m}A_0(n/k_{\rho})}o_P(1).
\end{aligned}
$$
It follows that as $N\to \infty$,
$$
\sqrt{k_{\rho}m}A_0(n/k_{\rho})\suit{T_{k_{\rho},\tau}-3\frac{\rho-1}{\rho-3}}= -\gamma\frac{2(1-\rho)^2}{\rho }\suit{P_N^{(1)}-P_N^{(2)}}-\gamma \frac{6(1-\rho)^3}{\rho^2-3\rho}\suit{P_N^{(2)}-P_N^{(3)}}+O_P(1).
$$
Theorem \ref{Theorem of rho} is thus proved by applying the Cram\'er's delta method.
\end{proof}

\begin{proof}[Proof of Theorem \ref{Theorem : gamma}]
	By Proposition \ref{theorem for expansion}, as $N 
	\to \infty$, $R_{k_n}^{(1)}$ has the following asymptotic expansion:
	$$
	\sqrt{k_nm}\suit{R_{k_n}^{(1)}-\gamma}-\gamma P_N^{(1)}-\frac{g(k_n,n,\rho)}{1-\rho}\sqrt{k_nm}A_0(n/k_n)=o_P(1),
	$$
	which leads to 
	$$
	\sqrt{k_nm}\set{\suit{R_{k_n}^{(1)}}^2-\gamma^2}-2\gamma^2 P_N^{(1)}-2\gamma\frac{g(k_n,n,\rho)}{1-\rho}\sqrt{k_nm}A_0(n/k_n)=o_P(1).
	$$
Together with the asymptotic expansion of $R_{k_n}^{(2)}$, we have that, as $N\to \infty$,
$$
\sqrt{k_nm}\set{R_{k_n}^{(2)}-2\suit{R_{k_n}^{(1)}}^2 }-\gamma^2\suit{P_N^{(2)}-4P_N^{(1)}}-\sqrt{k_nm}A_0(n/k_n) g(k_n,n,\rho)\frac{2\gamma\rho}{(1-\rho)^2}=o_P(1).
$$
Thus, as $N\to \infty$,
$$
\begin{aligned}
	&\sqrt{k_nm}\suit{\tilde{\gamma}_{k_n,k_{\rho},\tau}-\gamma}
	\\
	&\quad =\sqrt{k_nm}\suit{R_{k_n}^{(1)}-\gamma}- \frac{1}{2R_{k_n}^{(1)}\hat{\rho}_{k_{\rho},\tau} (1-\hat{\rho}_{k_{\rho},\tau})^{-1}}\sqrt{k_n m}\set{R_{k_n}^{(2)}-2\suit{R_{k_n}^{(1)}}^2}\\
	&\quad = \gamma P_N^{(1)} +\sqrt{k_nm}A_0(n/k_n)\frac{g(k_n,n,\rho)}{1-\rho}+o_P(1)\\
	&\quad \quad -\frac{1}{2R_{k_n}^{(1)}\hat{\rho}_{k_{\rho},\tau} (1-\hat{\rho}_{k_{\rho},\tau})^{-1}}\set{\gamma^2\suit{P_N^{(2)}-4P_N^{(1)}}+\sqrt{k_nm}A_0(n/k_n) g(k_n,n,\rho)\frac{2\gamma\rho}{(1-\rho)^2}+o_P(1)}\\
	&\quad = \gamma P_N^{(1)}-\frac{\gamma^2(1-\hat{\rho}_{k_{\rho},\tau})}{R_{k_n}^{(1)}\hat{\rho}_{k_{\rho},\tau}}\suit{P_N^{(2)}/2-2P_N^{(1)}}\\
	&\quad \quad +\sqrt{k_n m}A_0(n/k_n)\frac{\rho}{(1-\rho)^2}g(k_n,n,\rho) 	\suit{\frac{1-\rho}{\rho}-\frac{1-\hat{\rho}_{k_{\rho},\tau}}{\hat{\rho}_{k_{\rho},\tau}}}+o_P(1).
\end{aligned}
$$
The relation $k_n/k_{\rho}\to 0$ implies that $A(n/k_n)/A(n/k_{\rho})\to 0$ as $N \to \infty$. Thus, by Theorem \ref{Theorem of rho}, we have that, as $N \to \infty$,
$$
\sqrt{k_n m}A_0(n/k_n)\frac{\rho}{(1-\rho)^2}g(k_n,n,\rho)\suit{\frac{1-\rho}{\rho}-\frac{1-\hat{\rho}_{k_{\rho},\tau}}{\hat{\rho}_{k_{\rho},\tau}}}=o_P(1).
$$
Together with the consistency of $ \hat{\rho}_{k_{\rho},\tau}$ and $R_{k_n}^{(1)}$, we have that, as $N \to \infty$,
$$
\sqrt{k_nm}\suit{\tilde{\gamma}_{k_n,k_{\rho},\tau}-\gamma}= \frac{\gamma}{\rho}\set{P_N^{(2)}(\rho-1)/2+P_N^{(1)}(2-\rho)}+o_P(1).
$$
Combining with Proposition  \ref{theorem for expansion}, we obtain that, as $N\to \infty$,
$$
\sqrt{k_nm}\suit{\tilde{\gamma}_{k_n,k_{\rho},\tau}-\gamma} \stackrel{d}{\to}  N\bra{0,\gamma^2 \set{1+\suit{\rho^{-1}-1}^2}}.
$$
\end{proof}

\bibliographystyle{apalike}
\bibliography{mybib}

\end{document}